\documentclass[11pt]{article}

\usepackage{endnotes}
\let\footnote=\endnote

\usepackage{natbib}
\bibpunct[, ]{(}{)}{,}{a}{}{,}%
\usepackage{url}
\usepackage{nicefrac}
\usepackage{amssymb}
\usepackage{amsmath,amsopn}
\usepackage{amsthm}
\usepackage{amsfonts}
\usepackage{amssymb}
\usepackage{hyperref}
\hypersetup{
	colorlinks,
	linkcolor={red!50!black},
	citecolor={blue!50!black},
	urlcolor={blue!80!black}
}
\usepackage{breakurl}
\newcommand{\LeftEqNo}{\let\veqno\@@leqno}
\usepackage{environ}
\usepackage{xcolor}
\usepackage{float}
\usepackage{graphicx}
\usepackage[noend]{algpseudocode}
\usepackage{algorithm}
\usepackage{tikz}
\usepackage{float}
\usepackage{array}
\usepackage{enumerate}
\usepackage{subcaption}
\usepackage{booktabs}
\usepackage{comment}
\usepackage{stackrel}
\usepackage{color,soul}
\usepackage{algorithmicx}
\usepackage{graphicx}
\usepackage{xspace}

\usepackage{footnote}
\usepackage{geometry}

\usepackage{natbib}
\bibpunct[, ]{(}{)}{,}{a}{}{,}%

\usepackage{authblk}
\usepackage{lipsum}

\makeatletter
\renewcommand\footnoterule{%
	\kern-3\p@
	\hrule\@width \textwidth
	\kern2.6\p@}
\makeatother

\makeatletter
\renewcommand*{\@fnsymbol}[1]{\ensuremath{\ifcase#1\or *\or \dagger\or \ddagger\or **\or \mathsection\or \mathparagraph\or \|\or  \dagger\dagger
		\or \ddagger\ddagger \else\@ctrerr\fi}}
\makeatother

\usepackage{hyperref}	 % MUST be the LAST loaded package 
\hypersetup{
	colorlinks={true},linkcolor={magenta},citecolor=blue!60!black,urlcolor  = blue!60!black,
}

\newtheorem{theorem}{Theorem}
\newtheorem{proposition}{Proposition}
\newtheorem{definition}{Definition}
\newtheorem{corollary}{Corollary}

\newtheorem{lemma}{Lemma}
\DeclareMathOperator*{\argmin}{argmin}

\algnewcommand{\LineComment}[1]{\State \(\triangleright\) #1}

\newcommand{\LTSP}{\textbf{LTSP}\xspace}
\newcommand{\LTSPR}{\textbf{LTSP-R}\xspace}
\newcommand{\OLTSP}{\textbf{O-LTSP}\xspace}

\newcommand{\SSS}{\texttt{SSS}\xspace}
\newcommand{\GS}{\texttt{GS}\xspace}
\newcommand{\FGS}{\texttt{FGS}\xspace}
\newcommand{\NFGS}{\texttt{NFGS}\xspace}
\newcommand{\LogNFGS}{\texttt{LogNFGS}\xspace}
\newcommand{\LTFS}{\texttt{LTFS}\xspace}
\newcommand{\ELTFS}{\texttt{LTFS+}\xspace}
\newcommand{\Replan}{\texttt{Replan}\xspace}
\newcommand{\ReplanSSS}{\texttt{Replan+SSS}\xspace}
\newcommand{\ReplanGS}{\texttt{Replan+GS}\xspace}
\newcommand{\ReplanFGS}{\texttt{Replan+FGS}\xspace}
\newcommand{\ReplanLogNFGS}{\texttt{Replan+LogNFGS}\xspace}

\newcommand{\timehorizon}{\mathcal{H}\xspace}

\newcommand{\file}{\ensuremath{f}}

\newcommand{\objectivevalue}{\ensuremath{v}}

\newcommand{\files}{\ensuremath{\mathcal{F}}}
\newcommand{\size}{\ensuremath{s}}
\newcommand{\nrequests}{\ensuremath{n}}
\newcommand{\lengthtape}{\ensuremath{m}}
\newcommand{\setrequests}{\ensuremath{\mathcal{R}}}
\newcommand{\request}{\ensuremath{r}}
\newcommand{\releasetime}{\ensuremath{t}}
\newcommand{\minibatches}{\ensuremath{\mathcal{B}}}
\newcommand{\conflictminibatch}{\ensuremath{\bar{\mathcal{B}}}}

\newcommand{\minibatch}{\ensuremath{b}}

\newcommand{\fileleft}{\ensuremath{l}}
\newcommand{\fileright}{\ensuremath{r}}

\newcommand{\tape}{\ensuremath{H}\xspace}

\begin{document}
	%%%%%%%%%%%%%%%%

\title{Theoretical and Practical Aspects of the Linear Tape Scheduling Problem}

\author[1]{Carlos Cardonha\thanks{carloscardonha@br.ibm.com}}
\author[1]{Lucas C. Villa Real\thanks{lucasvr@br.ibm.com}} 
\affil[1]{\small IBM Research}
\maketitle

\begin{abstract}	
		Magnetic tapes have been playing a key  role as means for storage of digital data for decades, and their unsurpassed cost-effectiveness still make them the technology of choice in several industries, such as media and entertainment. Tapes are mostly used for cold storage nowadays, and therefore the study  of scheduling algorithms for read requests tailored for these devices has been largely neglected  in the literature. In this article, we investigate the Linear Tape Scheduling Problem (\LTSP), in which read requests associated with files stored on a single-tracked magnetic tape should be scheduled in a way that the sum of all response times are minimized.  \LTSP has many similarities with classical combinatorial optimization problems such as the Traveling Repairmen Problem and the Dial-a-Ride Problem restricted to the real line; nevertheless,  significant  differences on structural properties and strict time-limit constraints of  real-world scenarios make \LTSP challenging and interesting on its own. In this work, we investigate several properties and algorithms for~\LTSP and some of its extensions. The results allowed for the identification of 3-approximation algorithms for~\LTSP and  efficient exact algorithms for some of its special cases. We also show that~\LTSPR, the version of the problem with heterogeneous release times for requests, is NP-complete. \OLTSP, the online extension of~\LTSPR, does not admit $c$-competitive algorithms for any constant factor~$c$, but we nevertheless introduce an algorithm for the problem and show through extensive  computational experiments on synthetic and real-world datasets that different embodiments of the proposed strategy are computationally efficient and over-perform by orders of magnitude an algorithm being currently used by real-world tape file systems.%
\end{abstract}

	\section{Introduction}

	A modern magnetic tape is a storage device in which files are distributed sequentially in a strip of plastic film. The reading/writing head is in a fixed position, so the tape drive needs to rewind or fast-forward the tape to the leftmost bit of the requested file in order to retrieve its content. 
	%After reaching the target position, the head needs to be moved up or down to select the correct track before the requested operations can execute. 
	Because access in tapes is sequential, response times %for write and read requests 
	depend on the tape's position and on the order in which operations are executed (\cite{Sandsta99}). 
	%Tape drive movements take place simultaneously with read or write operations; the velocity with which the tape moves while streaming data is the same independently of which of these two operations is being conducted. 
	Write requests are often handled asynchronously by storage systems through the use of caching memories, though, so the Linear Tape Scheduling Problem (\LTSP) is focused on the scheduling of read requests, that is, on the identification of an execution order for these requests that minimizes the \emph{overall response times} (i.e.,  the sum of all response times).
	
	% Differently from more recent storage technologies, the access to information in magnetic tapes is ``mechanic''; namely, in a reading or writing operation, the tape drive's head must be positioned in the location of the strip where the target file is located, an operation that might take minutes depending on how long is the strip and on how fast the XXX moves. 
	
	In scenarios where read requests occur frequently and involve files which are far apart in the tape, sequential access can seriously compromise performance; in these cases, hard disk drives and solid-state drives provide much faster response times and are therefore preferable. Nevertheless, certain aspects of magnetic tapes  make them the storage technology of choice in many industries. First, magnetic tapes still have the best price per bit ratios in the market, making them suitable  for ``cold storage'' (or long-term archiving), i.e., for storage of data whose access frequency is expected to be very low. Cold storage  is common practice in the media and entertainment industry, where high-resolution digital content is stored but may remain archived for many years without ever being read~(\cite{Frank12}).  In the oil industry, sonars aboard vessels often capture more data from deep oceans than the available computers can process, resulting in the archival of that data for eventual processing -- akin to situations faced by large scientific organizations~(\cite{Murray12}) and renowned projects such as the Large Hadron Collider~(\cite{Cavalli10}). Tapes are also optimized for streaming operations (sequential data transfer speeds may reach around 300MB/second in some scenarios), which makes them even more interesting for the industries described above.  In contrast with scenarios in which cold storage data is seldomly accessed, many organizations are opening access to their archives through the Internet, leading to unpredictable access patterns on sequential storage mediums. Moreover, batch processes are routinely found to be accessing data from said storage ~(\cite{Adams12}). For other tendencies in uses of magnetic tapes, see e.g., \cite{Parkash13}, \cite{VillaReal15}, and \cite{Pease10}.

	A few assumptions are made in the versions of~\LTSP investigated in this work. First, we do not consider overwrite, modify, and deletion operations. This assumption actually has a practical motivation: as data stored within a partition of a magnetic tape cannot be updated without causing data loss to the tracks nearby, such operations are undesirable and typically avoided whenever possible. 
	%Each job~$\request$ in~$\setrequests$  has a \textit{type}, given by function $\requesttype:\setrequests \rightarrow \{r,w\}$;  $\requesttype(\request) = w$ if~$\request$ is a \textit{write request}, and $\requesttype(\request) = r$ if~$j$ is a \textit{read request}. 
	As mentioned previously, we also do not consider write requests, as they are executed according to a FIFO policy and their response times are actually irrelevant in practice.
	%; namely, the execution of a write job becomes urgent only after the arrival of a read job associated with the same file. 
	%In order to explore this aspect, tape systems use a \textit{buffer area} to temporarily store files associated with write requests, postponing hence the execution of these requests. We assume in this work that the buffer has infinite capacity.  
	%Finally, there is no enforced policy dictating how read requests should be processed, and this is the point where a scheduling strategy plays a role in \LTSP. Informally, we can say that \LTSP is about the identification of scheduling plans that deliver low \textit{overall response times} (or the sum of the response times) to read requests. 
	Another assumption is the restriction to single track tapes, %scenarios where the tape consists of a single track, 
	that is, the head remains at a fixed vertical position while the tape moves to the left and to the right; although strong, this assumption reflects situations where batches of requests are relatively ``local'' (i.e., files belonging to the same working set stay close to each other in the tape), a desired configuration that has motivated tape partitioning schemes proposed in the literature~(see~\cite{Oracle11}). Finally, we assume that the velocity of the tape is always constant; we remark, though, that tapes are mechanical devices and therefore acceleration and momentum do have an impact.%have an impact on the velocity in which the request is served, 
	%though, so this is  a simplification of the problem. 

	%The demand for long-term archiving of digital data can be easilly observed in the media and entertainment industry. Often, high-resolution digital content produced needs to be stored for potential replay in the future (e.g., a replay of a soccer game), but many times that content remains archived for many years without ever being read~\cite{Frank12}. In the oil industry, sonars aboard vessels often capture more data from the deep oceans than the available computers can process, resulting in the archival of that data for eventual processing -- akin to situations faced by large scientific organizations~\cite{Murray12} and renowned projects such as the Large Hadron Collider~\cite{Cavalli10}. Finally, we remark that tapes are also being suggested as components of distributed computing-based storage solutions~\cite{Parkash13} and being used for online reading and writing of data files~\cite{VillaReal15} just like more popular removable media types~\cite{Pease10}.

	Work on the literature dedicated to \LTSP and similar problems is scarce; in the operations research literature, we identified only one paper addressing a similar problem~(\cite{Day1965letter}). One explanation for this is the limited computational power and the strict time restrictions constraining scheduling algorithms for the problem in  practical settings. Namely, scheduling plans are expected to be produced in less than one second (two hundred milliseconds is the rule-of-thumb), which may hinder the application of traditional operations research techniques for discrete optimization problems in medium- and large-scale scenarios. For instance, the file system of a regular desktop  machine may have more than 500,000 files, thus leading to situations which are very unlikely to be properly addressed with integer linear programming and constraint programming formulations.
	%For instance, in our computational experiments, we employed real-world instances extracted from  an imaging application in precision agriculture and from the traces of an anti-virus tool. For these scenarios, one may easily identify instances containing between 50,000 and 500,000 files, thus making them unlikely to be solvable with simple integer linear programming and constraint programming formulations.

	In this article, we investigate several aspects of the Linear Tape Scheduling Problem and some of its variations and special cases, thus largely extending the first article directly addressing the problem~(\cite{cardonha2016online}). Several theoretical results are presented, including structural properties, complexity and competitive analysis, and algorithms. For~\LTSP, we present 3-approximation algorithms for the general case and show that it can be solved exactly in polynomial time if all files have the same size and the number of requests per file is at most one. We also show that~\LTSPR, the extension of~\LTSP in which requests may have different release times, is NP-hard. \OLTSP, the online version of~\LTSPR,  does not admit $c$-competitive algorithms for any given constant~$c$. Nevertheless, we present an algorithm for \OLTSP and show through computational experiments, with synthetic data and with instances extracted from  an imaging application in precision agriculture, that this strategy overperforms  the one currently used by implementations of the Linear Tape File System widely used in the industry~(\cite{ISO20919}). 
	
	%sFirst, new theoretical aspects of \LTSP are presented; in particular, we show that the offline \LTSP with release times is NP-complete, thus strengthening the connection of \LTSP with \textbf{TRP} in the line.  We also introduce a novel constraint programming formulation for the problem, which was designed on top of structural properties of \LTSP that have been presented by~\cite{cardonha2016online} as well as others that were first introduced in the present work. Computational results involving instances of the anti-virus problem are presented and compared against the performance of the reference implementation of LTFS.

	The article is organized as follows. In Section~\ref{sec:literature}, we present a literature overview and show the connections between~\LTSP and several other combinatorial optimization problems. We present the notation that will be used throughout the text and define formally~\LTSP, \LTSPR, and \OLTSP in Section~\ref{sec:description}. The theoretical results and the algorithms are presented in Section~\ref{sec:theory}. Finally, we report on our computational experiments in Section~\ref{sec:experiments} and present our conclusions  in Section~\ref{sec:conclusion}.

	\section{Related work}\label{sec:literature}

	The Linear Tape Scheduling Problem (\LTSP) is closely related to  the Traveling Repairman Problem (\textbf{TRP}), a variant of the Traveling Salesman Problem (\textbf{TSP}) in which requests are generated on  vertices of a graph and the goal is to visit each vertex and execute its associated (repair) request while minimizing the sum of waiting times (\cite{AfratiCPPP86}).  
	%Total waiting times are the focus of the so-called latency problems (see e.g., \cite{PaepeLSSS04}); these problems are typically hard, but strong topological assumptions, such as customer placement in a line, may be sufficient to make the problem computationally tractable (see e.g., \cite{coene2011charlemagne} for an example on the Periodic Latency Problem). 
	\textbf{TRP} has also been studied under  different names, such as Delivery Man Problem or Minimum Latency Problem (see e.g., \cite{fischetti1993delivery}, \cite{simchi1991minimizing}, \cite{coene2011charlemagne}). \cite{AfratiCPPP86} showed that  \textbf{TRP} is NP-complete in the general case and that \textbf{Line-TRP}---a version of the problem in which all requests are distributed in a line (i.e.,  in~$\mathbb{R}^1$)---can be solved in polynomial time if requests do not have deadlines and their servicing times are either insignificant or  equal for all machines. Recently, \cite{Bock15} showed that \textbf{Line}-\textbf{TRP} with general processing times and deadlines is strongly NP-complete; the author also presented an overview of current knowledge on the computational complexity of several versions of the problem (and of \textbf{Line}-\textbf{TSP} as well, in which the goal is to minimize the tour duration); in particular, the complexity of  \textbf{Line}-\textbf{TRP} without time windows (i.e., without release times and due dates) and with general processing times is still unknown. Differently from~\textbf{TRP}, though, in~\LTSP the tape points to a different position after the execution of a read operation, as a file must be completely traversed (from the left to the right)  to be read. This aspect leads to a subtle but significant difference between \LTSP and \textbf{Line-TRP}  with general processing times and without time windows.% (i.e., without release times and due dates for jobs).  

	\textbf{LTSP} is also similar to the Dial-a-Ride Problem, in which fleets of vehicles are used to transport products between vertices in a graph. An extensive classification of Dial-a-Ride problems is presented by~\cite{PaepeLSSS04}; in particular, they showed that the minimization of the sum of all completion times is NP-hard even in scenarios where there is just one vehicle of capacity one (i.e., the vehicle carries at most one product at any point in time) servicing all requests and all vertices are positioned in a line. \LTSP is equivalent to the special case of the Dial-a-Ride Problem in a line with a single vehicle  in which the trajectories described by each route are pairwise non-intersecting, the sources are always on the left of their associated destinations, and the vehicle starts to the right of the rightmost destination. The reduction employed by~\cite{PaepeLSSS04} (based on the Circular Arc Coloring Problem) cannot be used in this special case, leaving hence its complexity  open.

	The Linear Tape Scheduling Problem is related to the Single-Vehicle Scheduling Problem	on a Line (\textbf{Line-VSP}), in which  the goal is to minimize the completion time of jobs associated with vertices, while eventually taking into account  release time, deadlines, handling (or service) time, and initial/final position of the vehicle. The version with handling times and time windows is NP-complete (\cite{GareyJ79}). In scenarios where service times are equal to zero and only release times are considered, the problem can be solved in polynomial time~(\cite{PsaraftisSMK90}). \cite{KarunoNI98} presented a polynomial-time 1.5 approximation algorithm for \textbf{Line-VSP}  with release and handling times; later, \cite{KarunoN01}  showed that the problem admits a Polyonomial-Time Approximation Scheme (PTAS) if the number of vehicles is fixed and larger than one. For other results, see e.g., \cite{Tsitsiklis92} and~\cite{KarunoN05}.

	\LTSP can be formulated as the Time-Dependent Traveling Salesman Problem (\textbf{Time-Dependent TSP}) if all files are associated with the same number of requests, as penalties can be evaluated based on the %level (i.e., 
	visit order of the vertices. 
	%Namely, for each file~$\file$ associated with at least one request, the constructed ATSP graph~$D$ has one associated vertex~$v_{\file}$; additionally, $D$ contains an artificial vertex~$v_a$, which represents the rightmost position of the tape. Graph~$D$ is complete, i.e., there is an arc connecting each pair of vertices in the graph. A path in~$D$ is associated with a sequence of files being executed in the original \LTSP instance, and as each file with pending requests needs to be serviced, feasible solutions for both problems are equivalent. The cost of arc~$(v,v')$ depends on whether the associated operation in the \LTSP instance is taking place on Phase 1 or Phase 2, and this is the reason why time zones play a role on the definition of their values; for Phase 1 arcs, the cost of $(v_{\file},v_{\file'})$ the $k$-th arc of the tour is given by $(n-k)c( \fileright(\file') - \fileright(\file)   )$ if $\file < \file'$ and $(n-k)c( \fileright(\file) - \fileright(\file') + 2\size(\file') )$ otherwise.
	\cite{abeledo2013time, abeledo2010time} and~\cite{allahverdi1999review} showed the close relationship between the \textbf{Time-Dependent TSP} and~$1/ s_{ij}/\sum_j{C_j}$, the machine scheduling problem in which there is just one machine, setup times apply between the execution of tasks, and the objective is to minimize the sum of completion times (for more information about the notation of theoretical machine scheduling problems, see e.g., \cite{lawler1993sequencing}). In its more general setting, though, \LTSP is actually more closely related to $1/ s_{ij}/\sum_j{w_jC_j}$, the extension with weighted completion times. From the theoretical and algorithmic perspective, though, we were not able to identify results in the literature that could be properly leveraged and/or adapted to~\LTSP while exploring its unique tape-like structure.

	\LTSP is the version of the  problem in which the head is initially positioned over the end of the tape (i.e., the rightmost bit of the last file) and all requests are released at time~0.  One generalization is the Linear Tape Scheduling Problem with Release Times (\LTSPR), in which each request has a known release time, which defines  a lower bound on the time the associated file should be read for that particular request to be serviced. More interesting from a practical perspective, though, is \OLTSP, the online extension of~\LTSPR in which release times are unknown and decisions need to be taken (and executed) in real-time, before all requests become available. 
	%In many practical settings, though, one typically has to deal with the online version of~\LTSPR, in which release times are unknown and decisions need to be taken (and executed) in real-time, before all requests become available.
	\OLTSP is related to the Online Traveling Salesman Problem (\textbf{OLTSP}), a variation of the Traveling Salesman Problem in which new vertices to be visited are informed  during  the tour. If write requests are also considered, the resulting problems is related to the \textbf{Homing-OLTSP}, defined by~\cite{AusielloFLST01}, as the end of the tape needs to be visited from time to time. In the same work, the authors presented a $1.75$-competitive algorithm for the \textbf{Homing-OLTSP}  in a line and showed that no online algorithm can be better than 1.64-competitive for the problem. Handling times are zero for the \textbf{Homing-OLTSP}, but~\cite{Augustine02} showed that any $c$-competitive online algorithm for the problem with  zero handling times yields a $(c+1)$-competitive online algorithm for non-negative handling times. Recently, \cite{bjelde2017tight} closed the competitiveness gap of \textbf{OLTSP}  (i.e., the authors presented a 1.64-competitive algorithm for the problem) and presented new complexity results for several versions of the  offline Dial-A-Ride problem.

	Despite its practical relevance and its connections with several classical combinatorial optimization problems, \LTSP and its variants have not received significant attention so far in the operations research literature. The closest problem we identified in the literature is the Request Problem,
	%which was published by \citeauthor{Day1965letter}, who investigated and presents an integer linear programming formulation of the Request Problem, 
	in which one wishes to retrieve records that may be retrieved from several multi-record files distributed across a magnetic tape~(\cite{Day1965letter}). In the Request Problem, the goal is to minimize the time spent traversing the selected multi-record files (the time spent by the tape to move from one file to other is ignored), so it is basically a variation of the Set Covering Problem.
	%One possible explanation for this gap in the literature may be the limited computational power and the strict time restrictions constraining scheduling algorithms for the problem in  practical settings. Namely, scheduling plans are expected to be produced in less than one second (two hundred milliseconds is the rule-of-thumb), a limitation that hinders the application of traditional operational research techniques for discrete optimization problems. 
	Apart from this article,
	%For this reason, 
	most related work simply sort the list of requests according to their offset relative to the current position of the tape (\cite{Schaeffer11,Zhang06}). More complex approaches incorporate the cost of re-positioning the tape on different tape models, but they rely on time-consuming characterizations of the physical tapes and on access to low level hardware information, and not on optimized scheduling algorithms tailored for the problem~(\cite{Sandsta99,Hillyer96}). The present article extends the work of~\cite{cardonha2016online}, which was focused on~\LTSP and~\OLTSP; one difference, however, is that we do not consider the scheduling of read operations in  online scenarios.

	%\LTSP was initially introduced by~\cite{cardonha2016online}. In this work, structural properties of the offline \LTSP are investigated and employed in the design of algorithms for its online variant. For the offline \LTSP, the computational complexity of the problem remained open, but a 3-approximation algorithm was obtained. The online \LTSP was shown to be hard from the perspective of competitive analysis, in the sense that the problem does not admit a $c$-competitive online algorithm even if preemption is allowed. 

	%\hl{OR references:}
	%\begin{itemize}
	%	%\item Letter to the Editor—On Optimal Extracting from a Multiple File Data Storage System: An Application of Integer Programming %(https://pubsonline.informs.org/doi/abs/10.1287/opre.13.3.482)
	%	\item Technical Note—Optimization of the Capacity in a Storage System %https://pubsonline.informs.org/doi/abs/10.1287/opre.19.2.544
	%\end{itemize}

	\section{Problem description}\label{sec:description}

	Let~\tape be a single-track magnetic tape and~$\files = (\file_1,\file_2,\ldots,\file_n)$ be the set (or sequence) of \textit{files} stored in~$\tape$; we assume that~$\files$ does not change over the planning horizon. 
	%We overload notation and use $\files = (\file_1,\file_2,\ldots,\file_n)$ to indicate how the files are ordered in the tape, with $\file_1$ being the tape's leftmost file; we employ $\file + 1$ ($\file - 1$) to denote the file located to the right (left) of $\file$ in the tape.
	The content of a file~$\file$ in~$\files$ may be retrieved through the submission of a (read) \emph{request} (or \emph{job}) to~$\tape$.
	% End-users and software services  may retrieve the content of a file~$\file$ in~$\files$ by sending a \emph{requests} (or \emph{jobs}) to~$\tape$. 
	The set of requests is denoted by~$\setrequests$, and each~$\request$ in~$\setrequests$ is associated with exactly one file~$\file(\request)$ in $\files$. The set of requests associated with file~$\file$ is denoted by $\setrequests(\file) \subseteq \setrequests$, with $\nrequests(\file) = |\setrequests(\file)|$; in certain scenarios, $\nrequests(\file(\request))$ may also be interpreted as the \textit{weight} of a single request for~$\file(\request)$.  We extend this notation and also use $\setrequests(\files') = \bigcup_{\file \in \files'}\setrequests(\file)$ to denote requests associated with files in~$\files' \subseteq \files$. $\tape$ has a single read head, and the \textit{position} of~$\tape$ refers to the bit on which the head is currently located.  An \textit{interval} $(l,r) \in \mathbb{N}^2$ denotes the region of the tape starting at block~$l$ and finishing at block~$r$.

	A file must be \textit{read} in order to have its associated requests \textit{executed} (or \textit{serviced}). In order to read~$\file$, the tape must move
	%a read request on file~$\file$, the tape needs to \textit{read}~$\file$, an operation that takes place whenever the tape moves 
	from the leftmost bit of its leftmost block, denoted by~$\fileleft(\file)$,  to the rightmost bit of its rightmost block~$\fileright(\file)$; movement direction matters, so files traversed from the right to the left cannot be read. The \textit{size} of~$\file$ is given by~$\size(\file) = \fileright(\file) - \fileleft(\file) + 1$; different files may have different sizes. For $\file_1,\file_2 \in \files$, we say that $\file_1$ is on the \textit{left} of~$\file_2$ in~$\files$ (and that $\file_2$ is on the \textit{right} of~$\file_1$) if $\fileleft(\file_1) < \fileleft(\file_2)$. We employ $\file + 1$ ($\file - 1$) to denote the file located immediately to the right (left) of $\file$ in the tape, whose ordering is defined by sequence~$\files$.
	%;  so $\file_i < \file_j$ if $i < j$ for $f_i,f_j \in \files$. 
	%For write requests, the head needs to be on the leftmost unoccupied bit of the tape (or, more informally, in the \textit{end of~\tape}), denoted by~$\lengthtape$. 

	We assume that all operations take place in a  discrete \textit{time-horizon}~$\timehorizon \subseteq \mathbb{N}$; each element of~$\timehorizon$ is referred to as a \textit{time-step}. The velocity with which the tape moves is constant and does not change if a read operation is being performed;
	%One important characteristic of \LTSP  is the fact that execution times and tape movements are ``merged'', since a read operation in a file necessarily takes the tape to the leftmost bit of the next file.
	without loss of generality, we assume that the tape traverses one block per time-step. 
	%We also assume that the leftmost file in~$\files$ is associated with at least one request in~$\setrequests$.
	%velocity is always equal to 1, i.e., the tape drive needs~$s(\file)$ \textit{time-steps} to traverse file~$\file$, and that 
	% and that all files are positioned side by side starting from the beginning of the tape, that is, there is no free block to the left of any occupied bit. 
	%We overload notation and use $\files = (\file_1,\file_2,\ldots,\file_n)$ to indicate how the files are ordered in the tape, with $\file_1$ being the tape's leftmost file; we employ $\file + 1$ ($\file - 1$) to denote the file located to the right (left) of $\file$ in the tape. 
	The \textit{length} of the tape equals the position of its rightmost unoccupied block, denoted by~$\lengthtape$, so we have  $\lengthtape = \sum_{\file \in \files}\size(\file)$. In all versions of the Linear Tape Scheduling Problem investigated in this work, we assume that~$\lengthtape$ is the initial position of~$\tape$; this assumption reflects the fact that tapes are optimized for write operations and should therefore be always ready to execute them as fast as possible whenever requested.

	In situations where the time-step~$i$ in which decisions are being made or actions are being executed  can be inferred from the context (e.g., the moment when the leftmost block of a given file is reached for the first time given a certain schedule), we may refer to  \textit{pending requests}, which are elements in~$\setrequests$ that have not been serviced yet (i.e., requests whose planned  service times are larger than~$i$). The waiting time of a request~$\request$ finishes as soon as~$\file(\request)$ starts to be read, and
	%, that is, by the time the tape moves from the leftmost bit of file~$\file(\request)$ to the right. 
	all pending requests associated with a file are serviced simultaneously when the file is traversed.
	
	%A request can only be serviced after being released.  %Namely, in scenarios where release times are non-trivial (i.e., they are not necessarily all equal to zero), 
	Function $\releasetime: \setrequests \rightarrow \mathbb{N}$ defines the \textit{release time} of each request~$\request$, which is a lower bound for the execution of~$\request$;  if file~$\file$ is read at time-step~$i$ but there is some request~$\request$ in~$\setrequests$  such that~$\file(\request) = \file$ and $\releasetime(\request) > i$, $\file$ needs to be traversed again for~$\request$ to be serviced. %Release times are generated within a  discrete \textit{time-horizon}~$\timehorizon \subseteq \mathbb{N}$; each element of~$\timehorizon$ is referred to as a time-step.
	In scenarios where function~$\releasetime$ is not explicitly defined, we assume that $\releasetime(\request) = 0$ for every request~$\request$ in~$\setrequests$.

	Formally, the basic version of the Linear Tape Scheduling Problem is defined as follows:
	\begin{definition}[Linear Tape Scheduling Problem (\LTSP)] Given a set of files~$\files$ and a set of requests~$\setrequests$, %a type function $\requesttype$, and a size function $\size$, 
		schedule all requests  in a way that the sum of all response times  is minimized.
	\end{definition}
	
	We  also investigate~\LTSPR, the extension of~\LTSP which allows non-trivial release times.

	\begin{definition}[Linear Tape Scheduling Problem with Release Times (\LTSPR)] Given a set of files~$\files$, a set of requests~$\setrequests$,  
		%a type function $\requesttype$, 
		%a size function $\size$, 
		and a function $\releasetime$ defining release times, schedule all requests in a way that release times are respected and the sum of all response times is minimized.
	\end{definition}
	
	Finally, we also consider~\OLTSP, the online extension of~\LTSPR. The only information made available for algorithms addressing this version of the problem is the set of files and the current position of the tape. Set~$\setrequests$ is completely unknown a priori, with information about any individual request~$\request$ being made available only after its release in time-step~$\releasetime(\request)$.

	\begin{definition}[Online Linear Tape Scheduling Problem (\OLTSP)] Given a set of files~$\files$, a set of requests~$\setrequests$, and a function~$\releasetime$ defining release times whose values are unknown beforehand, 
		%a type function $\requesttype$, 
		%a size function $\size$, and a release time function~$\releasetime$ which also defines the time at which the request is informed to the algorithm,  
		schedule all requests in a way that the sum of all response times is minimized.
	\end{definition}
	
	%In the offline versions of~\LTSP, we assume that all jobs are read requests. In the online \LTSP, release times of  jobs are unknown and~$\setrequests$ contains both write and read requests. Moreover, preemption is not allowed, that is, once a file associated with pending jobs starts to be traversed, its end must be reached. 

	\section{Theoretical results and algorithms}\label{sec:theory}
	
	We present in this section several structural properties, complexity results, and algorithms for \LTSP, \LTSPR, and~\OLTSP. Early versions of some of these results (indicated in the article) and of algorithms \SSS, \GS, and \FGS appeared in~\cite{cardonha2016online}.

	%Some of the results originally shown by~\cite{cardonha2016online} are presented here without proofs for the sake of completeness. 
	%We start with~\LTSP and afterwards we proceed with the analysis of~\LTSPR and~\OLTSP.
	
	\subsection{Linear Tape Scheduling Problem}
	
	In the Linear Tape Scheduling Problem, the tape is initially positioned at the end of the tape (i.e., position~$\lengthtape$) and all jobs are released at time~0.
	Optimal scheduling plans for~\LTSP can be divided in two phases.  In \textit{Phase 1}, the tape is traversed from its end (position~$\lengthtape$) to its beginning~(position~$1$), i.e., from the right to the left. This movement may be interrupted for the execution of one or more \textit{mini-batches}; similarly to files, a mini-batch~$\minibatch$ has a leftmost block~$\fileleft(\minibatch)$, a rightmost block~$\fileright(\minibatch)$, and a size~$\size(\minibatch) = \fileright(\minibatch) - \fileleft(\minibatch) + 1$. The execution of a mini-batch consists of a \textit{read movement}, in which the tape moves from~$\fileleft(\minibatch)$ to~$\fileright(\minibatch)$, and a \textit{return movement}, from $\fileright(\minibatch)$ to $\fileleft(\minibatch)$. 
	%A mini-batch~$\minibatch'$ cannot be executed between the end of the execution movement and the start of the return movement of some other mini-batch~$\minibatch$; in these situations, we say that all movements are associated with the same mini-batch. 
	%We extend the notation used for files and jobs to mini-batches; that is, a mini-batch~$\minibatch$ has a leftmost bit~$\fileleft(\minibatch)$, a rightmost bit~$\fileright(\minibatch)$, and a size~$\size(\minibatch) = \fileright(\minibatch) - \fileleft(\minibatch)$. 
	We say that file~$\file$ belongs to mini-batch~$\minibatch$ (or, similarly, that~$\minibatch$ contains~$\file$) if $\fileleft(\minibatch) \leq \fileleft(\file) < \fileright(\file) \leq \fileright(\minibatch)$. The set of files traversed by~$\minibatch$ is denoted by~$\files(\minibatch)$, and if $|\files(\minibatch)| = 1$, we say that~$\minibatch$ is an \textit{atomic mini-batch}. The leftmost and the rightmost blocks of a mini-batch~$\minibatch$ coincide with the leftmost and the rightmost bits of (non necessarily different) files~$\file$ and~$\file'$ in~$\files$, respectively, in any optimal solution of \LTSP, so we may also use pair~$(\file,\file')$ to represent~$\minibatch$.
	%; to see this, note that, by substituting each mini-batch~$\minibatch$ not adhering to this rule for the largest mini-batch~$\minibatch'$ properly contained in~$\minibatch$ that satisfies both conditions, one obtains a new set~$\minibatches_1'$ such that $v(\minibatches_1') \leq v(\minibatches_1)$. 
	We denote by~$\minibatches$ the set of all possible mini-batches, i.e.,  $\minibatches = \{  (\file,\file') \in \files \times \files: l(\file) \leq l(\file')  \}$, and by~$\minibatches_1$ the set of mini-batches selected for execution in Phase 1. We use~$\files(\minibatches_1)$ to denote the set of files traversed by at least one mini-batch in~$\minibatches_1$, i.e., $\files(\minibatches_1) = \bigcup_{\minibatch \in \minibatches_1}\files(\minibatch)$, and~$\conflictminibatch_{\minibatch}$ to denote the set of mini-batches~$\minibatch'$ in~$\minibatches$ such that~$\files(\minibatch')$ intersects but does not contain all files traversed by~$\minibatch$, i.e., 
	$\conflictminibatch_{\minibatch} = \{ \minibatch'   \in \minibatches: 
	\files(\minibatch) \cap \files(\minibatch') \neq \emptyset \text{ and } 
	\files(\minibatch) \setminus \files(\minibatch') \neq \emptyset\}$.

	In \textit{Phase 2}, the tape moves from position~$1$ to position~$\lengthtape$. In any optimal solution for \LTSP, all pending jobs associated with file~$\file$ are serviced (simultaneously) whenever~$\file$ is traversed from the left to the right, so the tape never moves to the left in \textit{Phase 2}. Note that the beginning of Phase 2 always coincides  with the execution of the leftmost file~$\file$ such that $\nrequests(\file) > 0$; we can assume without loss of generality that this file is always the leftmost file of the tape, and note that this file never belongs to~$\files(\minibatches_1)$.

	\subsubsection{Structural results}
	
	We start with simple structural results that partially characterize optimal solutions for~\LTSP and will be useful in the upcoming sections.  First, we show that, given~$\minibatches_1$, finding the optimal order of execution is straightforward, so the optimization of~\LTSP is equivalent to the identification of an optimal set~$\minibatches_1$. Therefore,  we use $v(\minibatches_1)$ to denote the objective value (sum of response times for all requests) of the schedule defined by set~$\minibatches_1$. 
	\begin{proposition}\label{prop:offline-has-simple-schedule}
		For every optimal schedule of~$\LTSP$, there is an equivalent solution in which mini-batches are executed in the (decreasing) order defined by their leftmost bits.
	\end{proposition}
	\begin{proof}
		Let~$\minibatches_1$ be an optimal set of mini-batches, $\minibatch$ and~$\minibatch'$ be mini-batches in~$\minibatches_1$ such that $\fileleft(\minibatch) < \fileleft(\minibatch')$, and assume that $\minibatch'$ is scheduled to be executed immediately after $\minibatch$.
		If all requests associated with files in~$\files(\minibatch')$ were serviced by the time~$\minibatch'$ starts to be executed, $\objectivevalue(\minibatches_1 \setminus \{\minibatch'\}) \leq \minibatches_1$, so $\minibatch'$ can be safely removed from the set of mini-batches. Otherwise, 	 
		%If $\minibatch' \in \conflictminibatch_{\minibatch}$, $()$ 	the sub-optimality of the solution is clear, so we can assume that $\files(\minibatch) \cap \files(\minibatch') = \emptyset$.
		%	Let us suppose for the sake of contradiction that this condition does not hold.
		%	Then there exists a pair of mini-batches~$\minibatch,\minibatch'$ in~$\minibatches_1$ such that $\fileleft(\minibatch) < \fileleft(\minibatch')$ and $\minibatch'$ is executed immediately after $\minibatch$. 
		by replacing $\minibatch$ and $\minibatch'$ for $\minibatch'' = \left(\fileleft(\minibatch),\fileright(\minibatch')\right)$  in~$\minibatches_1$ and in the execution sequence of mini-batches, one would obtain a new set~$\minibatches_1'$ and an associated execution plan that would avoid having the interval $(\fileleft(\minibatch), \fileleft(\minibatch'))$ being unnecessarily traversed twice before the execution of~$\minibatch'$, thus reducing response times for pending requests associated with files in~$\files(\minibatch')$ and contradicting the optimality of~$\minibatches_1$.
		%(and in $\bigcup_{\file}$ with all files that would be serviced after the execution of~$\minibatch'$). 
%		%$\qed$
	\end{proof}
	
	\begin{corollary}\label{one-minibatches-per-file}
		No pair of mini-batches share the same leftmost file.
	\end{corollary}
	\begin{proof} Apply the same arguments used in the proof of Proposition~\ref{prop:offline-has-simple-schedule} assuming  $\fileleft(\minibatch) = \fileleft(\minibatch')$ %and
		%replacing~$\minibatch$ and ~$\minibatch'$ for
		and using $\minibatch'' = \left(\fileleft(\minibatch), \max(\fileright(\minibatch),\fileright(\minibatch')) \right)$.	%$\qed$
	\end{proof}
	\begin{corollary}\label{minibatches-start-on-first-pass}
		Every mini-batch starts from a position that is being visited for the first time. 
	\end{corollary}

	\begin{proposition}\label{minibatches-requested-extremities}
		The leftmost and the rightmost files covered by a mini-batch are associated with at least one request each. 
	\end{proposition}
	\begin{proof} 	Let~$\minibatch \in \minibatches_1$ be a mini-batch for which these conditions do not hold. By replacing~$\minibatch$ for the largest mini-batch~$\minibatch' = (\file,\file')$ such that $\files(\minibatch') \subseteq \files(\minibatch)$ and $\nrequests(\file)\nrequests(\file') > 0$,  we obtain a new schedule~$\minibatches_1' = (\minibatches_1 \setminus \{\minibatch'\}) \cup \{(\file,\file')\}  $ such that $\objectivevalue(\minibatches_1') \leq  \objectivevalue(\minibatches_1)$; moreover, if $\minibatch$ is not the leftmost mini-batch in  $\minibatches_1$, the inequality is strict, as at least one request will have its response time reduced by $2(\fileleft(\minibatch) - \fileleft(\minibatch'))$.	%$\qed$
	\end{proof}
	\begin{corollary}\label{minibatches-finish-on-untouched-file}
		The rightmost file of every  mini-batch~$\minibatch \in \minibatches_1$ is traversed for the first time when~$\minibatch$ is executed. 
	\end{corollary}

	\begin{proposition}\label{minibatches-do-not-intersect}
		For every instance of~\LTSP, there is an optimal solution~$\minibatches_1$ which does not contain simultaneously mini-batch~$\minibatch$ and any element in $\conflictminibatch_{\minibatch}$.
	\end{proposition}
	\begin{proof}
		Let us assume that~$\minibatch \in \minibatches_1$ and  let~$\minibatch'$ be an element of~$\conflictminibatch_{\minibatch}$ such that $\minibatch' < \minibatch$. As all requests associated with files in~$\files(\minibatch)$ will have been serviced by the time~$\fileleft(\minibatch)$ is reached for the third time (i.e., after the execution of~$\minibatch$ and during the read movement of~$\minibatch'$), we do not reduce response times by letting the head traverse interval $(\fileleft(\minibatch),  \fileright(\minibatch'))$,
		% during the execution of~$(\file_1,\file_2)$, 
		so~$\minibatch'$  %$(\file_1,\file_2)$ 
		can be replaced for~$(\fileleft(\minibatch'),\minibatch - 1)$, where $\minibatch -1$ is the file located to the left of~$\minibatch$, 
		without increasing the sum of the response times. %Finally, if $\minibatch$ is not the leftmost mini-batch in~$\minibatches_1$, response times will decrease with this substitution.
		%$\qed$
	\end{proof}
	%\begin{corollary}\label{minibatches-finishes-on-first-pass}
	%	The rightmost file of any mini-batch is being traversed for the first time. 
	%\end{corollary}

	%\begin{corollary}\label{cor:novoidminibatchend}
	%	In any optimal scheduling plan for the offline \LTSP, every mini-batch starts and ends at files associated with at least one pending request. 
	%\end{corollary}

	We conclude  with two essential results  for some of the algorithms presented in this article. 
	
	\begin{proposition}[\cite{cardonha2016online}, Proposition 6]\label{prop:executeinphase1}
		Given an arbitrary set~$\minibatches_1$,
		$v(\minibatches_1 \cup \{ (\file,\file )  \} )  <  v(\minibatches_1)$ for
		file~$\file \in \files \setminus \files(\minibatches_1)$ if
		\begin{eqnarray}\label{ineq:executeinphase1}
		\nrequests(\file)\left(\fileleft(\file) +  \sum_{ \substack{\minibatch \in \minibatches_1\\ \minibatch < \file }}\size(\minibatch)  \right) 
		>
		\size(\file)\left(    \sum\limits_{  \substack{ \file' \in  \files \\ \file' < \file}   }\nrequests(\file')   + 
		\sum\limits_{  \substack{ \file' \in  \files \setminus \files(\minibatches_1) \\ \file' > \file}   }\nrequests(\file')        \right).		
		\end{eqnarray}
	\end{proposition}
	\begin{proof}
		We show that if file~$\file$ satisfies the conditions above but does not belong to~$\files(\minibatches_1)$, then $\minibatches_1$ can be improved with the inclusion of the atomic mini-batch containing~$\file$.
		The expression in the left-hand side of  Inequality~\ref{ineq:executeinphase1}  gives the reduction on the waiting times of all requests in~$\setrequests(\file)$ if~$\file$ is executed in Phase 1 as an atomic mini-batch. The expression considers not only the distance between~$\file$ and the beginning of the tape (given by~$\fileleft(\file)$), but also the time spent with the execution of mini-batches starting on the left of~$\file$; note that this distance is actually traversed twice. The expression in the right-hand side accounts for the increase by~$2\size(\file)$ on the the waiting times of all requests pending by the time~$\fileleft(\file)$ is reached for the first time.
		%	All files located to the left of~$\file$ have their response times increased by~$2\size(\file)$. Moreover, files located to the right of~$\file$ which have not been serviced yet by the time~$\fileleft(\file)$ is reached for the first time will also have their response times increased by~$2\size(\file)$. 
		Therefore, if~$\file \notin \files(\minibatches_1)$ and Inequality~\ref{ineq:executeinphase1} holds, $v(\minibatches_1 \cup \{ (\file,\file )  \} )  <  v(\minibatches_1)$  and the result holds.%$\qed$
	\end{proof}
	\begin{corollary}\label{prop:executeinphase2}
		Given an arbitrary set~$\minibatches_1$ and a mini-batch~$(\file,\file)$ in $\minibatches_1$,
		$v(\minibatches_1 \setminus \{ (\file,\file )  \} )  <  v(\minibatches_1)$ if
		%The atomic mini-batch $\minibatch \in \minibatches_1$  containing file~$\file$ should be removed from~$\minibatches_1$ if
		%	Let~$\minibatches_1$ be the set of mini-batches selected for execution in Phase 1 in an optimal solution for an instance of the offline \LTSP, and let~$\files_1$ be the set of files belonging to the mini-batches in~$\minibatches_1$. The execution of file~$\file$ is not postponed to Phase 2 if  
		\begin{eqnarray}\label{ineq:executeinphase2}
		\nrequests(\file)\left(\fileleft(\file) +  \sum_{ \substack{\minibatch \in \minibatches_1\\ \minibatch < \file }}\size(\minibatch)  \right) 
		<
		\size(\file)\left(    \sum\limits_{  \substack{ \file' \in  \files \\ \file' < \file}   }\nrequests(\file')   + 
		\sum\limits_{  \substack{ \file' \in  \files \setminus \files(\minibatches_1) \\ \file' > \file}   }\nrequests(\file')        \right).		
		%	\size(\file)\left(  \sum\limits_{  \substack{  \file' \in \bigcup\limits_{ \substack{ \minibatch \in \minibatches_1 \\ \minibatch > \file  }  }\minibatch   } }\nrequests(\file')   +  \sum_{ \substack{\file' \in \files\\ \file' < \file }}\nrequests(\file')        \right).
		%	\size(\file)\left(  \sum\limits_{  \substack{  \file' \in \bigcup\limits_{ \substack{ \minibatch \in \minibatches_1 \\ \minibatch > \file  }  }\minibatch   } }\nrequests(\file')   +  \sum_{ \substack{\file' \in \files\\ \file' < \file }}\nrequests(\file')        \right).		
		\end{eqnarray}
		%	$\fileleft(\file)\nrequests(\file) > \size(\file)|\setcomplementrequests(\file)|$. 
	\end{corollary}

	\subsubsection{Algorithms and approximability}
	
	The minimization of the makespan (i.e., the time by which all requests are serviced) is trivial for \LTSP; an optimal solution for the problem with this optimization criteria consists of moving the head from position~$\lengthtape$ directly to the leftmost bit of the tape, without executing mini-batches (i.e., $\minibatches_1 = \emptyset$), and then returning back to the end of the tape while servicing all requests. Note that the same strategy  also minimizes the distance (or the number of blocks) traversed by the tape. 
	\cite{KarunoNI98} show that this strategy, to which they refer as Simple Schedule Strategy (\SSS), yields a 1.5-approximation algorithm  for the Single-Vehicle Scheduling Problem on a line with release and handling times. The proposition below shows that this strategy may deliver arbitrarily bad plans for the~\LTSP.

	\begin{proposition}[\cite{cardonha2016online}, Proposition 3]\label{prop:simplebad}
		%The simple schedule strategy delivers arbitrarily bad plans in worst-case scenarios of the offline LTSP.
		\emph{\SSS} is not a $c$-approximation for \LTSP for any constant~$c$.
	\end{proposition}
	\begin{proof}
		Let us assume without loss of generality 
		that~$\lengthtape$ is a perfect square integer. Let~$\files = (\file_1,\file_2,\ldots,\file_{\sqrt{\lengthtape}})$ be such that $\fileleft(\file_1) = \size(\file_1) = 1$; $\fileleft(\file_2) = 2$ and $\size(\file_2) = \lengthtape - \sqrt{\lengthtape}$; and $\fileleft(\file_i) = \lengthtape - \sqrt{\lengthtape} + i - 1$ and $\size(\file_i) = 1$ for each $i$ in $[3,\sqrt{\lengthtape}]$. Finally, let $|\setrequests(\file_i)|$ be $1$ if $i \neq 2$ and 0 otherwise.
		%	, and let us assume that~$\setrequests$ contains exactly one read request associated with each file in~$\mathcal{F} \setminus \{\file_2\}$. 
		
		In the simple schedule strategy, files
		%the tape's head moves to position~$1$ and returns to position~$m$, executing~
		$\file_1, \file_2, \ldots, \file_{\sqrt{\lengthtape}}$ are executed in the order defined by their indices. Asymptotically, the resulting overall request time yielded by \SSS is
		%the simple schedule strategy for this family of scenarios is %approximately
		\begin{eqnarray*}
			O\left(\lengthtape(\sqrt{\lengthtape})\right) + 
			O\left(\sqrt{\lengthtape}(\lengthtape - \sqrt{\lengthtape})\right) + O\left(\frac{\sqrt{\lengthtape}\sqrt{\lengthtape}}{2}\right);
		\end{eqnarray*}
		first, all $O(\sqrt{\lengthtape})$ jobs wait until the tape moves from position~$\lengthtape$ to~$1$, thus resulting in a penalty of $O(\lengthtape\sqrt{\lengthtape})$. After the execution of~$\file_1$, the remaining $\sqrt{\lengthtape}-2$ requests wait for the tape to move forward $O(\lengthtape - \sqrt{\lengthtape})$ positions, resulting in an additional penalty of $O\left(\lengthtape\sqrt{\lengthtape} - \lengthtape\right)$, before being sequentially serviced, with an aggregate penalty given by $\sqrt{\lengthtape} + (\sqrt{\lengthtape}-1)+\ldots+1 = O(\frac{\sqrt{\lengthtape}\sqrt{\lengthtape}}{2})$. Therefore, the overall response time delivered by \SSS is~$O\left(\lengthtape\sqrt{\lengthtape}\right)$. 
		
		Alternatively, by executing mini-batch~$(\fileleft(\file_2),\fileright(\file_{\sqrt{\lengthtape}}  ))$,  
		%~$\minibatch = (\lengthtape - \sqrt{\lengthtape} + 2, \lengthtape)$, 
		the waiting time of all requests except one is reduced by a factor of~$\sqrt{\lengthtape}$, thus yielding an overall request time
		%thus leaving only the job associated with~$f_1$ to be executed during Phase 2, 
		asymptotically equal to
		\begin{eqnarray*}
			O\left(\sqrt{\lengthtape}\sqrt{\lengthtape}\right) + O\left(\frac{\sqrt{\lengthtape}\sqrt{\lengthtape}}{2}\right) + O(\lengthtape), 
		\end{eqnarray*}
		which is $O(\lengthtape)$. The ratio between the penalties of both strategies for this family of scenarios converges asymptotically to $O(\sqrt{\lengthtape})$, and since~$\lengthtape$ can be made arbitrary large, we conclude that the simple schedule strategy is not $c$-competitive for \LTSP for any constant~$c$. %$\qed$
		%simple schedule strategy delivers arbitrarily bad plans in worst-case scenarios of the offline LTSP.
	\end{proof}

	\SSS is rather myopic, in the sense that all requests are left for execution only in Phase 2. An ``opposite'' approach would consist of having one atomic mini-batch in Phase 1 for each file~$\file$ in the tape such that~$\nrequests(\file) > 0$ (with the exception of the file positioned on the leftmost position of the tape, as its execution defines the start of Phase 2). We refer to this policy as the Greedy Strategy (\GS), and below we show that \GS is a 3-approximation algorithm for~\LTSP.
	
	\begin{proposition}[\cite{cardonha2016online}, Proposition 4]\label{prop:greedy3approx}
		\GS is a 3-approximation for~\LTSP and the factor is tight.
	\end{proposition}
	\begin{proof}
		%For any instance of~\LTSP, %given~$\mathcal{J}$, 
		First, note that a lower bound for the response time of any job~$\request$ in $\setrequests$ is given by
		% = \{j_1,\ldots,j_n\}$ 
		%be the set of requests associated with an arbitrary instance of the offline LTSP. A lower bound for the response time of job~$j$ in $\mathcal{J}$ is given by
		\begin{eqnarray*}\label{lbprop}
			\lengthtape - \fileleft(\request) = \sum_{\substack{\file' \in \files\\ \fileleft(\file') > \fileleft(\request) \text{ and } \nrequests(\file') > 0 }}\size(\file') + \sum_{\substack{\file' \in \files \\ \fileleft(\file') > \fileleft(\request) \text{ and } \nrequests(\file') = 0 }}\size(\file');
		\end{eqnarray*}
		the left-hand side of the equality indicates that this  bound is equal to the distance between~$\fileleft(\request)$ and the end of the tape.
		%, i.e., this is the \hl{region} of the tape located to the right of~$\file(\request)$. 
		This region of the tape can be decomposed in two parts, as indicated in the right-hand side of the equation; one containing files associated with requests,	which is represented by the first term, and the other containing files without pending requests.

		In the greedy strategy, all requests associated with files located to the right of~$\file(\request)$ are serviced in atomic mini-batches before~$\fileleft(\request)$ is reached; consequently, the response time of~$\request$ is given by
		\begin{eqnarray}\label{exactgreedy}
		\sum_{\substack{\file' \in \files \\ \fileleft(\file') > \fileleft(\request) \text{ and } \nrequests(\file') = 0 }}\size(\file') +  
		\sum\limits_{\substack{\file' \in \files \\ \fileleft(\file') > \fileleft(\request) \text{ and } \nrequests(\file') > 0  }}3\size(\file'),
		\end{eqnarray}
		since the head of the tape visits each of these files exactly three times before moving to~$\fileleft(\request)$. These facts hold for every~$\request \in \setrequests$, so it follows that \GS is a 3-approximation for~\LTSP. 
		%	and since~$z(\request) \geq 0$, it follows that the greedy strategy is a 3-approximation for  the offline LTSP. Below, 
		
		Let us assume now that $|\files| = 2$, with $\fileleft(\file_1) = \size(\file_1) = 1$, $\fileleft(\file_2) = 2$, $\size(\file_2) = \lengthtape-1$, $n(\file_1) = \lengthtape-1$, and $n(\file_2) = 1$, that is, the tape has two files, with $\lengthtape-1$ requests associated with the first and 1 with the second. The overall request time delivered by \GS for this family of scenarios is  $\lengthtape(\lengthtape-1) + 2(\lengthtape-1)^2 + m - 1 = 3\lengthtape^2 - 4\lengthtape - 3$, whereas the plan delivered by \SSS yields an overall response time of $\lengthtape^2 + 1$; asymptotically, the ratio between both solutions converge to~$3$, thus showing that this approximation factor is tight. %$\qed$
	\end{proof}
	
	Note that the family of instances used in the proof of Proposition~\ref{prop:greedy3approx} also shows that GS does not necessarily yield better schedules than \SSS.

	The Filtered Greedy Strategy (\FGS), presented in Algorithm~\ref{alg:filteredselectivegreedy}, extends \GS by applying Corollary~\ref{prop:executeinphase2} successively to the solution produced by \GS as a post-processing operation in order to improve the solution. Namely, \FGS proceeds by checking whether Inequality~\ref{ineq:executeinphase2} holds for some file~$\file$ in~$\files(\minibatches_1)$; if so, the associated atomic mini-batch is removed and the procedure is repeated until no more mini-batches can be removed by this criteria. The following lemma shows that the order with which non-atomic mini-batches are selected for remotion based on Corollary~\ref{prop:executeinphase2} is irrelevant. 
	
	\alglanguage{pseudocode}
	\begin{algorithm}[!ht]
		\caption{Filtered Greedy Strategy (\FGS)}
		\label{alg:filteredselectivegreedy}
		\begin{algorithmic}
			\Procedure{\FGS}{$\files,\setrequests,h,\minibatches_1$}
			\State $\minibatches_1 = \GS(\files,\setrequests,h)$
			%\LineComment{\emph{Infinite loop for online scenarios}}
			\For{$k \in \{1,\ldots,|\files|\}$}
			\For{$\file \in \files(\minibatches_1)$}
			\LineComment{\emph{Check condition of Corollary~\ref{prop:executeinphase2}}}
			\If{
				$\nrequests(\file)\left(\fileleft(\file) +  \sum\limits_{ \substack{\minibatch \in \minibatches_1\\ \minibatch < \file }}\size(\minibatch)  \right) 
				<
				\size(\file)\left(    \sum\limits_{  \substack{ \file' \in  \files \\ \file' < \file}   }\nrequests(\file')   + 
				\sum\limits_{  \substack{ \file' \in  \files \setminus \files(\minibatches_1) \\ \file' > \file}   }\nrequests(\file')        \right)$
			}
			\State $\minibatches_1 = \minibatches_1 \setminus \{(\file,\file)\}$
			\EndIf
			\EndFor
			\EndFor	
			\State return $\minibatches_1$
			\EndProcedure
		\end{algorithmic}
	\end{algorithm}
	\begin{lemma}\label{lemma:fgs_stable}
		If file~$\file$ satisfies Inequality~\ref{ineq:executeinphase2} in the $i$-th iteration of \FGS, then
		$\file$ also satisfies Inequality~\ref{ineq:executeinphase2} in the $j$-th iteration for every $j \geq i$.
	\end{lemma}
	\begin{proof}  Let~$w^l_{\file}(k)$ and~$w^r_{\file}(k)$ denote the value of the left-hand side and of the right-hand expressions of Inequality~\ref{ineq:executeinphase2} for file~$\file$ on the $k$-th iteration of \FGS, respectively. Suppose that $w^l_{\file}(k) < w^r_{\file}(k)$, and let $\minibatches_{(i,j)} = \{ \minibatch_i, \minibatch_{i+1}, \ldots, \minibatch_j \}$ be the set of (atomic) mini-batches removed between iterations~$i$ and~$j$. The remotion of each mini-batch in $\minibatches_{(i,j)}$ located to the left of~$\file$ decreases the penalty on the left-hand side of Inequality~\ref{ineq:executeinphase2} (i.e., there are less mini-batches to be executed, so the tape will reach $\fileleft(\file)$ faster), so we have $w^l_{\file}(j) \leq w^l_{\file}(i)$. Analogously, by excluding mini-batches located to the right of~$\file$, we increase the penalty on the right-hand side of Inequality~\ref{ineq:executeinphase2} (i.e., there are more requests that will only be executed in Phase 2), so  $w^r_{\file}(i) \leq w^r_{\file}(j)$ and the result holds. %$\qed$
	\end{proof}

	\begin{proposition}\label{prop:filteredgreedy3approx}
		\FGS is a 3-approximation for~\LTSP.
	\end{proposition}
	\begin{proof}
		Let~$\minibatches_1$ be the set of mini-batches obtained in the first step of Algorithm~\ref{alg:filteredselectivegreedy} (i.e., by \GS) and let~$\minibatches_1^*$ be the set returned by \FGS.  By construction, one mini-batch is removed per iteration of \FGS. Each individual operation reduces the current overall response time, so we have that \FGS dominates \GS, in the sense that~$v(\minibatches_1^*) \leq v(\minibatches_1)$ for every instance of the problem. From Proposition~\ref{prop:greedy3approx}, it follows that \FGS is a 3-approximation for~\LTSP.
		%$\qed$
	\end{proof}

	The previous algorithms employ only atomic mini-batches, so we investigate now the natural extension of this approach, in which non-atomic mini-batches are also considered. First, we introduce the concept of \textit{opportunity cost}~$\Delta(\minibatches_1,\minibatch)$, which reflects the impact on the overall waiting times caused by the inclusion of mini-batch~$\minibatch$ in~$\minibatches_1$, i.e.,
	$\Delta(\minibatches_1,\minibatch) = v(\minibatches_1 \cup \{\minibatch\}) - v(\minibatches_1)$. The opportunity cost~$\Delta(\minibatches_1,\minibatch)$ is evaluated at the moment when $\fileleft(\minibatch)$ is reached for the first time and is given by
	%always computed in comparison with~$\minibatches_1$.
	%The value of~$\Delta(\minibatches_1,\minibatch)$ is given by 
	
	\vspace{-5pt}
	\begin{scriptsize}
		\begin{eqnarray}\label{def:impact}
		\Delta(\minibatches_1,\minibatch) &=& 
		\size(\minibatch)\left(    
		\sum\limits_{  \substack{ \file \in  \files   \\ \file < \minibatch}   }\nrequests(\file)   + 	
		\sum\limits_{  \substack{ \file \in  \files \setminus \files(\minibatches_1 ) \\ \file > \minibatch}   }\nrequests(\file)        \right) +		
		\sum_{\file \in   \files(\minibatch) \setminus \files(\minibatches_1)  }  \left(\fileleft(\file) - \fileleft(\minibatch) \right)\nrequests(\file) 
		- \sum_{\file \in   \files(\minibatch) \setminus \files(\minibatches_1)  }  \nrequests(\file)   
		\left(\fileleft(\file) +  \sum_{ \substack{\minibatch' \in \minibatches_1\\ \minibatch' < \minibatch   }}\size(\minibatch')  \right) \nonumber  \\
		&=& \size(\minibatch)\left(    
		\sum\limits_{  \substack{ \file \in  \files   \\ \file < \minibatch}   }\nrequests(\file)   + 	
		\sum\limits_{  \substack{ \file \in  \files \setminus \files(\minibatches_1 ) \\ \file > \minibatch}   }\nrequests(\file)        \right) -		
		\sum\limits_{\file \in   \files(\minibatch) \setminus \files(\minibatches_1)  }   \nrequests(\file) \left(    \fileleft(\minibatch) +    
		\sum\limits_{ \substack{\minibatch' \in \minibatches_1\\ \minibatch' < \file   }}\size(\minibatch')  \right).
		\end{eqnarray}
	\end{scriptsize}
	%A non-atomic mini-batch~$\minibatch$ improves the solution described by a set~$\minibatches_1$ if
	%\begin{footnotesize}
	%\begin{eqnarray}\label{ineq:nonatomicequation}
	%\sum_{\file \in   \files(\minibatch) \setminus \files(\minibatches_1)  }  \nrequests(\file)   
	%\left(\fileleft(\file) +  \sum_{ \substack{\minibatch' \in \minibatches_1\\ \minibatch' < \file   }}\size(\minibatch')  \right) 
	%<
	%\size(\minibatch)\left(    
	%	\sum\limits_{  \substack{ \file \in  \files   \\ \file < \minibatch}   }\nrequests(\file)   + 	
	%	\sum\limits_{  \substack{ \file \in  \files \setminus \files(\minibatches_1 ) \\ \file > \minibatch}   }\nrequests(\file)        \right) +		
	% \sum_{\file \in   \files(\minibatch) \setminus \files(\minibatches_1)  }  \left(\fileleft(\file) - \fileleft(\minibatch) \right)\nrequests(\file)
	%\end{eqnarray}
	%\end{footnotesize}
	
	In short, by including~$\minibatch$, requests associated with files in $\files(\minibatch) \setminus \files(\minibatches_1)$ will have their waiting times reduced by the amount of time the tape will spend to return to position~$\fileleft(\minibatch)$ in Phase 2, whereas all the others will wait $\size(\minibatch)$ additional time-steps due to the execution of~$\minibatch$. %Note that we should also account for the additional amount of time spent by files in~$\files(\minibatch)$ which are not located in the leftmost position of~$\minibatch$ in order to estimate~$\Delta(\minibatches_1,\minibatch)$ precisely.
	We leverage the concept of opportunity cost to extend \FGS and derive the Non-atomic Filtered Greedy Strategy (\NFGS), a strategy that improves the quality of solutions created by \FGS based on the inclusion of non-atomic mini-batches; the procedure is formally presented in Algorithm~\ref{alg:nonatomicfilteredselectivegreedy}.

	\alglanguage{pseudocode}
	\begin{algorithm}[!ht]
		\caption{Non-atomic Filtered Greedy Strategy (\NFGS)}
		\label{alg:nonatomicfilteredselectivegreedy}
		\begin{algorithmic}
			\Procedure{\NFGS}{$\files,\setrequests,h,\minibatches_1$}
			\State $\minibatches_1 =  \FGS(\files,\setrequests,h,\minibatches_1)$
			%				\For{$k \in \{1,\ldots,|\files|\}$}
			%				\State $\minibatch^* = \argmin\limits_{\minibatch \in \minibatches \setminus \minibatches_1  } \Delta(\minibatches_1,\minibatch)$ 			
			%				\If{ $\Delta(\minibatches_1,\minibatch^*) < 0$}
			%				\State $\minibatches_1 = \minibatches_1 \cup \{\minibatch^*\}$
			%				\EndIf
			%				\EndFor
			\For{$\file \in \files$}
			\State $\minibatches_1' = \minibatches_1 \setminus \bigcup\limits_{\file' \in \files}(\file,\file') $
			\State $\file^* = \argmin\limits_{\file' \in \files, \file' > \file} \Delta(\minibatches_1', (\file,\file'))$
			\If{ $\Delta(\minibatches_1', (\file,\file^*)  ) < 0$}
			\State $\minibatches_1 = \minibatches_1' \cup \{(\file,\file^*)\}$
			\EndIf
			%		
			%		\For{$\file' \in \files, \file' > \file$}
			%		%\State $\minibatch^* = \argmin\limits_{\minibatch \in \minibatches \setminus \minibatches_1  } \Delta(\minibatches_1,\minibatch)$
			%		\State $\minibatches_1' = \minibatches_1 \setminus \bigcup\limits_{\file'' \in \files}(\file,\file'') $
			%		\If{ $\Delta(\minibatches_1', (\file,\file')  ) < 0$}
			%		\State $\minibatches_1 = \minibatches_1' \cup \{(\file,\file')\}$
			%		\EndIf
			%		\EndFor
			\EndFor
			\State return $\minibatches_1$
			\EndProcedure
		\end{algorithmic}
	\end{algorithm}

	\NFGS starts with the schedule produced by \FGS and iteratively tries to improve it by incorporating non-atomic mini-batches based on opportunity costs. Namely, in each iteration, the algorithm creates~$\minibatches_1'$, which is composed by all mini-batches in~$\minibatches_1$ that do not start from some given file~$\file$ (note that, from Corollary~\ref{one-minibatches-per-file}, there exists at most one such mini-batch); files are inspected in the main loop of~\FGS according to the order defined by~$\files$ (leftmost file is checked first). Given~$\minibatches_1'$, \FGS checks whether the smallest opportunity cost~$\Delta(\minibatches_1',(\file,\file'))$ among all files~$\file' > \file$ is negative; if so, $\minibatches_1$ is substituted for $\minibatches_1' \cup \{(\file,\file')\}$, which has a smaller objective value by construction. 
	An efficient implementation of \NFGS runs in time~$O(|\files|^3)$; namely, each of the $O(|\files|^2)$ potential mini-batches is tested once, and the update of data structures supporting fast evaluation of opportunity costs takes time~$O(|\files|)$. For practical applications, this running time may be prohibitively large, so we also considered in our experiments \LogNFGS, the version of \NFGS in which we only consider mini-batches~$(\file,\file')$ such that the relative positions of~$\file$ and~$\file'$ in the tape do not differ by more than $\log(|\files|)$.

	Finally, as~$\minibatches_1$ is substituted only for solutions with better objective values, 
	%identifies the non-atomic mini-batch~$\minibatch$ with minimum opportunity cost in each step. If $\Delta(\minibatches_1,\minibatch) < 0$, $\minibatch$ is included in~$\minibatches_1$, yielding a new solution~$\minibatches_1'$ with better objective value; otherwise, the algorithm may stop, as improvements are not achievable with the inclusion of a single mini-batch. From Corollary~\ref{minibatches-start-on-first-pass}, it follows that each file in~$\files$ can be the leftmost file of at most one mini-batch in~$\minibatches_1$, so the number of iterations of NFGS is~$O(|\files|^2)$. Finally, the number of non-atomic mini-batches is~$O(|\files|^2)$, so NFGS runs in time~$O(|\files|^3)$, and 
	the same arguments employed in the proof of Proposition~\ref{prop:filteredgreedy3approx} can be used to show the following result:
	\begin{corollary}\label{prop:nonatomicfilteredgreedy3approx}
		\NFGS and \LogNFGS are 3-approximations for~\LTSP.
	\end{corollary}
	
	We remark that the inclusion of mini-batches appears to be analytically more challenging than their remotion. Namely, whereas Lemma~\ref{lemma:fgs_stable} provides ``stability'' to \FGS, in the sense that the order with which decision are made is irrelevant and the output of the strategy will always deliver the same objective value, %-- regardless of the implementation being used -- 
	order does matter for \NFGS. The identification of an optimal remotion strategy for \NFGS seems to be as hard as solving~\LTSP exactly, a challenge that we briefly discuss below. Before, we investigate special cases of~\LTSP which have practical relevance and also provide additional structure that makes the problem solvable in polynomial time.

	\subsubsection*{Tractable scenarios}

	In the theorem below, we show that~\LTSP  can be efficiently solved if all files have the same size and  exactly one associated request each. 
	
	\begin{theorem}\label{thm:emptyoptimal}
		If  $\size(\file) = k$ for some given $k \in \mathbb{N}$ and $\nrequests(\file) = 1$ for every $\file$ in $\files$, $\minibatches_1 = \emptyset$ is optimal.
	\end{theorem}
	\begin{proof}
		
		The result follows from a few  properties of  this special case of~\LTSP. First, we show that we may restrict our attention to solutions consisting solely of  atomic mini-batches.
		
		\begin{lemma}\label{lemma:onlyatomic}
			For every~$\minibatches_1 \subseteq \minibatches$, there is a $\minibatches_1' \subseteq \minibatches$ such that $\objectivevalue(\minibatches_1') \leq \objectivevalue(\minibatches_1)$ and all mini-batches in~$\minibatches_1'$ are atomic.
		\end{lemma}
		\begin{proof}
			Let $\minibatches_1$ be an optimal solution containing a non-atomic mini-batch~$\minibatch = (\file,\file')$. % Let~$\file$ and~$\file'$ be the leftmost and the rightmost files in~$\files(\minibatch)$. 
			From Corollaries~\ref{minibatches-start-on-first-pass} and~\ref{minibatches-finish-on-untouched-file}, $\file$ and~$\file'$  will be traversed for the first time during the execution of~$\minibatch$. Let~$x$ be the difference between the time-steps in which positions~$\fileleft(\file)$ and~$\fileleft(\file')$ are visited for the first time according to~$\minibatches_1$, and let $x' = \fileleft(\file') - \fileleft(\file)$; note that $x \geq x' \geq k = \size(\file)$, $x \neq x'$ if and only if~$\minibatches_1$ contains one or more mini-batches $\minibatch'$ such that $\files(\minibatch') \subset \files(\minibatch)$, and  $x' \neq k$ if and only if $\file \neq \file'-1$. Let $u = (x'-k)/k$ be the number of files covered by~$\minibatch$ in addition to~$\file$ and~$\file'$.
			
			We compare $\objectivevalue(\minibatches_1)$ with $\objectivevalue(\minibatches_1')$, where $\minibatches_1' = (\minibatches_1 \setminus \minibatch) \cup (\file,\file'-1)$; this modification does not change the execution time of other mini-batches, so we restrict our attention to response times of requests in~$\setrequests(\minibatch)$ starting from the moment when~$\fileleft(\file')$ is being visited for the first time. 
			%First, we analyze the overall response time of requests in~$\setrequests(\file)$ and~$\setrequests(\file')$ starting from the moment when~$\fileleft(\file')$ is being visited for the first time. 
			For~$\minibatches_1$, this value equals $ (x) + (x + x') = 2x + x'$, where the first and seconds terms in the leftmost expression are the response times for requests in~$\setrequests(\file)$ and~$\setrequests(\file')$, respectively. For~$\minibatches_1'$, the waiting time of~$\setrequests(\file')$ vanishes and the response times equal 	$2k + x$, so the difference of these values from~$\minibatches_1$ to~$\minibatches_1'$ is  $ x + x' - 2k$. 
			%By including $(\file',\file')$ on $\minibatches_1$ and substituting~$\minibatch$ for $(\file,\file'-1)$, we obtain a new solution~$\minibatches_1'$ for which the overall response time for both files starting from the time~$\fileleft(\file')$ is visited for the first time reduces to that of requests in~$\setrequests(\file)$, which is equal to $2k + x. $
			The remaining~$u$ requests have their response times increased by~$2k$ each, so we have
			\[
			\objectivevalue(\minibatches_1) - \objectivevalue(\minibatches_1') =  x + x' - 2k - 2(x' - k) = x - x',
			\]
			As~$x \geq x'$,
			%		, the waiting time of each of the remaining~$u$ requests increase by~$2k$, so we have
			%		$\objectivevalue(\minibatches_1') \leq \objectivevalue(\minibatches_1)$, with strict inequality holding if $|\files(\minibatches_1)| > 2$. 
			%		
			%		If~$x > x'$, $|\files(\minibatch)| = 2 + (x'-k)/k$, i.e., $\minibatch$ covers $u = (x'-k)/k$ files in addition to~$\file$ and~$\file'$. The response times of each of these~$u$  additional requests increase  by~$2k$ from~$\minibatches_1$ to~$\minibatches_1'$, so we have 
			%		\[
			%		\objectivevalue(\minibatches_1) - \objectivevalue(\minibatches_1') =  x + x' - 2k - 2(x' - k) = x - x' > 0.
			%		\]
			one can always substitute~$\minibatches_1$ for~$\minibatches_1'$ without increasing the objective value by following the procedure above. 	As we have made no further assumptions about~$\minibatches_1$, 	
			%As the response times for all other pending requests remain unchanged,  $\objectivevalue(\minibatches_1') \leq \objectivevalue(\minibatches_1)$, and if $x > k$, the inequality is strict. 
			it follows that one can always generate an optimal solution containing only atomic mini-batches. %$\qed$
			%		\begin{itemize}
			%			\item  If $\file < \file'$ and both are solved in Phase 1, then $\file'$ should be executed first
			%			
			%			Consider waiting time starting from the moment when $\fileleft(\file')$ is being visited for the first time. If $\file$ is solved first, we have: 
			%			\begin{itemize}
			%				\item wait of $\file$: $k + x$ (traverse $\file$ plus files in between)
			%				\item wait of $\file'$: $2(k + x)$ (wait of $\file$ plus return time)
			%				\item total wait: $3k + 3x$
			%			\end{itemize}
			%			
			%			
			%			If $\file'$ solved first: 
			%			\begin{itemize}
			%				\item wait of $\file$: $k + x + 2k$
			%				\item wait of $\file'$: $0$
			%				\item total wait: $x + 3k$
			%			\end{itemize}
			%			
			%			\item If $\file$ and $\file'$ are consecutive, $x = 0$ and the result above shows that optimal solutions for this particular case of \LTSP may be replaced by other whose cost is not smaller and which consists only of atomic mini-batches.
			%			
			%			\item If $x > 0$, any solution containing a mini-batch traversing $\file$ and $\file'$ could be improved by solving $\file'$ first, in an atomic mini-batch.
			%		\end{itemize}
			%		%$\qed$
		\end{proof}	
		
		We say that a file~$\file$ is a \textit{hole} in~$\minibatches_1$ if $\file \notin \files(\minibatches_1)$ and if there exists a file~$\file'$ such that $\file' < \file$ and $\file' \in  \files(\minibatches_1)$, i.e., $\file$ is not scheduled to be traversed in Phase 1 according to~$\minibatches_1$ but there is at least one file positioned to its left that will. The lemma below shows that such configurations do not occur in optimal solutions.
		
		\begin{lemma}\label{lemma:noholes}
			Optimal solutions do not have holes.
		\end{lemma}
		\begin{proof}
			Let us assume for the sake of contradiction that~$\minibatches_1$ is optimal, consists solely of atomic mini-batches (by Lemma~\ref{lemma:onlyatomic}), and contains $h > 0$ holes. Let~$\file$ be the leftmost file in $\files(\minibatches_1)$, i.e.,  all holes appear to the right of~$\file$. If mini-batch~$(\file,\file)$ is removed from~$\minibatches_1$, the waiting time of its request increases by $2\fileleft(\file)$, whereas the waiting times of all other pending requests decrease by $2k(\fileleft(\file)/k + h)$. If $h > 0$, $\objectivevalue(\minibatches_1 \setminus \{(\file,\file)\}) < \objectivevalue(\minibatches_1)$, thus contradicting the optimality of~$\minibatches_1$. %$\qed$
		\end{proof}	
		
		Finally, the lemma below shows that optimal plans may always have their leftmost mini-batch removed without changes in the objective value. 
		\begin{lemma}\label{lemma:shrinkschedule}
			If $\minibatches_1$ is optimal,	$\objectivevalue(\minibatches_1 \setminus \minibatch) = \objectivevalue(\minibatches_1)$ if $\minibatch$ is the leftmost mini-batch in $\minibatches_1$.
		\end{lemma}
		\begin{proof}
			From Lemma~\ref{lemma:noholes}, we have that~$\minibatches_1$ contains no holes. Moreover, from Lemma~\ref{lemma:onlyatomic} we may assume that~$\minibatches_1$ contains only atomic mini-batches.
			If $\file$ is the file associated with the leftmost mini-batch~$\minibatch$ in $\minibatches_1$, we have $2k(\fileleft(\file)/k + h) = 2\fileleft(\file)$, so the reduction in the waiting time of pending requests is equal to the increase in the waiting time of the request in~$\setrequests(\file)$. Therefore, $\objectivevalue(\minibatches_1 \setminus \minibatch) = \objectivevalue(\minibatches_1)$ and the result holds.
			%$\qed$
		\end{proof}	
		The main result follows directly from the iterative application of Lemma~\ref{lemma:shrinkschedule} on any optimal solution for the problem consisting solely of atomic mini-batches..
		%$\qed$	
	\end{proof}

	%\begin{proposition}\label{prop:uniformfilesoneerequesteach}
	% If $\nrequests(\file) = 1$ for every $\file$ in $\files$, $\minibatches_1$ is optimal if and only if it does not contain intersecting mini-batches.
	%\end{proposition}
	%\begin{proof} 
	%	If $\minibatches_1 = \emptyset$, the service time of the request associated with the $i$-th file of the tape is equal to $|\files| + i$, thus yielding an overall cost $\objectivevalue(\emptyset) = |\files|^2 + |\files|(|\files| - 1)/2 = (|\files|^2 - |\files|)/2$. Our result is equivalent to the fact that every optimal solution $\minibatches_1'$ is such that $\objectivevalue(\emptyset) = \objectivevalue(\minibatches_1')$, so we show that $\objectivevalue(\minibatches_1) = \objectivevalue(\minibatches_1')$ if $\minibatches_1 \subseteq \minibatches_1'$, $|\minibatches_1' \setminus \minibatches_1| = 1$, and $\minibatches_1'$ does not contain intersecting mini-batches.
	%	
	%	Let~$\minibatches_1$ be such that $\objectivevalue(\minibatches_1) = \objectivevalue(\emptyset)$, let $\minibatch$ be a mini-batch in $\minibatches \setminus \minibatches_1$ such that $\files(\minibatch) \cap  \files(\minibatches_1) = \emptyset$, and let $p$ be the number of requests that will be pending by the time~$\fileleft(\minibatch)$ is reached for the first time according to~$\minibatches_1$. By including $\minibatch$, the response time of each file $\file$ in $\files(\minibatch)$ is reduced by $2\fileleft(\file)$ whereas the other files will have their response times increased by $\size(\minibatch)$.
	%	
	%	%$\qed$
	%\end{proof}

	Next, we show that optimal solutions for \LTSP also have a relatively simple structure (computable in polynomial time) if each file has at most one associated request.

	\begin{theorem}\label{thm:forcedholesoptimal}
		If  $\size(\file) = k$ for some $k \in \mathbb{N}$ and $\nrequests(\file) \leq 1$ for every $\file$ in $\files$, $\minibatches_1 = \{ (\file,\file): \file \in \files, \setrequests(\file) = 1 \}$ is optimal.
	\end{theorem}
	\begin{proof}
		
		First, note that Lemma~\ref{lemma:onlyatomic} also applies for this family of instances of \LTSP, so we can assume without loss of generality that all mini-batches composing optimal solutions are atomic.
		
		\begin{lemma}\label{lemma:forceexecution}
			If~$\minibatches_1$ is optimal and $\nrequests(\file) = 0$, $\file' \in \files(\minibatches_1)$ for every file~$\file'$ such that $\file' > \file$ and $\nrequests(\file') = 1$.
			%		, then $\files(\minibatches_1)$
			%		If $\nrequests(\file) = 0$, every file~$\file'$ such that $\file' > \file$ and $\nrequests(\file') = 1$ must be executed in Phase 1.
		\end{lemma}
		\begin{proof}
			
			Let us assume for the sake of contradiction that~$\minibatches_1$ is optimal, let~$\file$ be the rightmost file associated with one request which does not belong to~$\files(\minibatches_1)$, and let $e > 0$ be the number of files without requests positioned to the left of~$\file$. 
			
			By including~$\{\file,\file\}$ into~$\minibatches_1$, the waiting time of~$\setrequests(\file)$ gets reduced by $2\left( \fileleft(\file) + \sum\limits_{\file' \in \files(\minibatches): \file' < \file}\size(\file') \right)$, with the second term representing the time spent with mini-batches, whereas the overall response time of other requests pending by the time~$\fileleft(\file)$ is visited for the first time increases by $2k(l(f)/k - e)$.
			% for the other requests, where $e$ is the number of files to the left of~$\file$ without requests, whereas the waiting time of~$\file$ gets reduced by $2l(\file) + b$, where $b \geq 0$ is the time spent with schedule mini-batches taking place to the left of~$\file$. 
			If $e > 0$, we have $\objectivevalue(\minibatches_1 \cup \{\file\})  < \objectivevalue(\minibatches_1)$, so $\minibatches_1$ is sub-optimal and the result follows.
		\end{proof}
		
		Lemmas~\ref{lemma:onlyatomic} and~\ref{lemma:forceexecution} provide a direct characterization of optimal solutions for all files located to the right of the first file without pending requests. For the remaining files, it suffices to observe that the problem gets reduced to the one addressed in Theorem~\ref{thm:emptyoptimal} and that \GS also delivers an optimal solution in these cases, so the main result follows.%$\qed$
	\end{proof}
	
	Finally, despite their simplicity, we remark that these special cases of~\LTSP are of practical relevance. One example is seismic interpretation, in which specialists may wish to inspect one or more files containing segmented images of geological structures, which have roughly the same size, in order to check for analogous regions in different parts of the world.

	%We remark that the inclusion of mini-batches appears to be analytically more challenging than their remotion. Namely, whereas Lemma~\ref{lemma:fgs_stable} provides ``stability'' to FGS, in the sense that the order with which decisions are made is irrelevant and the output of the strategy will always deliver the same objective value -- regardless of the implementation being used -- order \emph{does} matter for N-FGS. The identification of an optimal remotion strategy for NFGS seems to be as hard as solving~\LTSP exactly, a challenge that we briefly discuss next.

	\subsubsection*{On the hardness of~\LTSP}

	We were not able to identify the  computational complexity of \LTSP, so a few comments (complementing the discussion presented in Section~\ref{sec:literature}) are in order. 
	%First, note that the tape's head moves from~$\fileleft(\file)$ to~$\fileright(\file)$ after reading file~$\file$, so execution operations and tape movements are actually merged. This aspect leads to a subtle but significant difference with  the Traveling Repairman Problem in a Line with general processing times and no time windows (note that, in this context, time windows refer to release times and due dates for jobs), which is probably the most similar problem to~\LTSP. Namely, in TRP, the repairman's position does not change after a node is visited and repaired, and the asymmetry of~\LTSP with respect to movements to the left and to the right stops the most immediate  reductions between these two problems from preserving the line structure. Additionally, we remark that the complexity of this version of TRP is  still unknown. \LTSP is also equivalent to the special case of the Dial-a-Ride Problem in a line with a single vehicle  in which the trajectories described by each route are pairwise non-intersecting, the sources are always on the left of their associated destinations, and the vehicle starts to the right of the rightost destination. The reduction employed by~\cite{PaepeL\SSS04}  (based on the Circular Arc Coloring Problem) cannot be used in this special case, leaving hence the complexity of this particular case of the problem open as well. 
	
	By inspecting more carefully the particular structure of~\LTSP, one can see that non-atomic mini-batches bring a recursive structure to the problem, in the sense that the concepts of Phase 1 and Phase 2 also apply to their execution and, in particular, that smaller mini-batches may be executed before the associated Phase 1. For instance, if mini-batch~$\minibatch$ is selected for execution and $\files(\minibatch) = \{\file_i,\file_{i+1},\file_{i+2} \}$, an optimal solution might also include the mini-batch $(\file_{i+1},\file_{i+1})$ (this would make sense in scenarios where~$\nrequests(\file_{i+1})$ is very large, for example).  Under this interpretation, a solution for the~\LTSP can be interpreted as the partition of the tape in up to~$|\files|$ components, where each part may be associated with a file that will only be executed in Phase 2, with an atomic mini-batch, or with a non-atomic mini-batch (which can be further divided using the same scheme, thus defining the recursive structure). We were not able to identify an exact polynomial-time combinatorial algorithm for~\LTSP by exploring this structure, though, and it is not clear whether such an algorithm exists. The main challenge seems to be interdependence across the selected blocks. First, the impact of  a mini-batch  depends on the mini-batches positioned to its right, as they define the number of pending requests. Additionally, the selection of a mini-batch is impacted by its duration and by the amount of time that the tape  will need to move towards the leftmost position of the mini-batch, so decisions made on the left also matter (intuitively, more mini-batches on the left make distances larger). The algorithms described above are ``myopic'', in the sense that they verify solely the inclusion or deletion of a single mini-batch in the decision process. Considering two or more mini-batches simultaneously could naturally improve the performance of these strategies, but running times become prohibitively large. Additionally, the number of possible solutions~$\minibatches_1$  is $O(2^{|\files|})$, and it is not clear whether straightforward memoization techniques (e.g., based on dynamic programming) can  avoid the inspection of an exponentially number of configurations.  Given these  challenges, % that need to be faced by combinatorial algorithms solving the problem, 
	we conjecture that~\LTSP is NP-hard.

	\subsection{Extensions with Release Times}
	
	The structural properties of the Linear Tape Scheduling Problem with Release (\LTSPR) times  are significantly different from those of~\LTSP. In particular, the scheduling of mini-batches becomes  more challenging; for instance, one can construct scenarios where a very large number of requests are release in the middle of the planning horizon in order to show that results such as Proposition~\ref{prop:offline-has-simple-schedule} and Corollaries~\ref{one-minibatches-per-file} and~\ref{minibatches-start-on-first-pass} do no hold for~\LTSPR. We show below that~\LTSPR %\LTSP with release times 
	is NP-complete.
	
	\begin{theorem}\label{thm:ltsprhard}
		\LTSPR is NP-complete.
	\end{theorem}
	\begin{proof}
		Let~$I$ be an instance of the Knapsack Problem with packing capacity $k$ and a set~$O$ of objects, with $v_o$ and~$w_o$ representing the value and the weight of object~$o \in O$, respectively.

		We construct an instance~$I'$ of~\LTSPR as follows. Each element of~$O$ is associated with a ``knapsack''  file~$f(o)$ in~$\files$ such that $s(f(o)) = w_o/2$ and $\nrequests(f(o)) = v_o$, i.e., the size and the number of requests of~$f(o)$ are related to the weight and to the value of~$o$, respectively.  $\files$ also contains a set~$\mathcal{Z}$ of ``artificial'' files such that $|\mathcal{Z}| = |O|-1$ and, for each $z \in \mathcal{Z}$, $\size(z) = k$ and~$\nrequests(z) = 0$. For every subset $\files'$ of $\files$, let~$O(\files')$ denote  the set of elements in~$O$ whose respective files belong to~$\files'$, i.e., $O(\files') = \{o \in O: f(o) \in \files'\}$, and for every $O' \subseteq O$, let $\files(O') = \{f(o) \in \files: o \in O'\}$. All requests associated with knapsack items are released at time-step zero, i.e., if $\request \in \setrequests(\files(O))$,  then $\releasetime_\request = 0$.
		Elements of $O(\files)$ and~$\mathcal{Z}$ appear in an interleaved manner in the tape, i.e., each artificial file has a knapsack file positioned on its left and another knapsack file on its right; in particular,  the rightmost file of the tape is a knapsack file. The order with which the elements of~$\files(O)$ are distributed in the tape is arbitrary, and so is the order of~$\mathcal{Z}$. Finally, set~$\files$ also contains a ``dummy''  file~$d$ positioned in the beginning of the tape (i.e., $l(d) = 1$); its size~$\size(d)$ and~$\nrequests(d)$ will be defined next. See Figure~\ref{fig:knapsacktape} for an illustration of this construction.
		We denote  the sum of the lengths of all files in $\files \setminus \{d\}$ by $W$; note that, by construction, $W = k(|O|-1) + \sum_{o \in O}w_o/2$.

		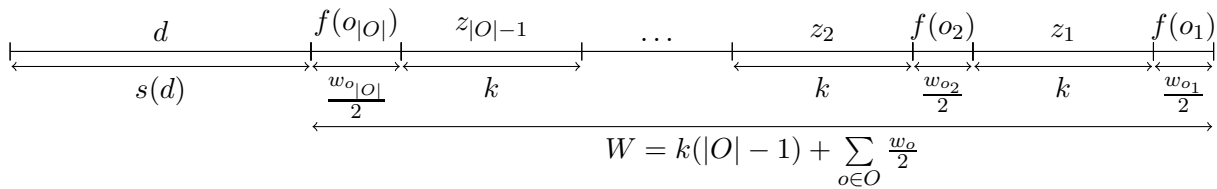
\begin{figure}[h!]
			\centering
			\usetikzlibrary{arrows,backgrounds,snakes}
			\begin{tikzpicture}[scale=0.8]%[font=\sffamily,\tiny]
			
			\newcommand\verticalBorders[2] {
				\draw (#1,-0.15) -- (#1,0.15); 
				\draw (#2,-0.15) -- (#2,0.15); 
			}
			
			\newcommand\labels[4] {
				\draw (#1,-0.15) -- (#1,0.15); 
				\draw (#2,-0.15) -- (#2,0.15); 
				\draw (#1,0) -- (#2,0) node[midway,above]{#3}; 
				\draw[<->] (#1+0.025,-0.25) -- (#2-0.025,-0.25) node[midway,below]{#4}; 
			}
			
			\newcommand\drawfile[4] {
				\verticalBorders{#1}{#2};
				\labels{#1}{#2}{#3}{#4};	
			}

			% Tape segment
			\draw[segment length=1cm] (0,0) -- (20,0);
			
			% File d
			\drawfile{0}{5}{$d$}{$\size(d)$};	
			
			% File o_{|O|}	
			\drawfile{5}{6.5}{$\file(o_{|O|})$}{$\frac{w_{o_{|O|}}}{2}$};
			
			% File z_{|O|}	
			\drawfile{6.5}{9.5}{$z_{|O|-1}$}{$k$};

			\draw (9.5,0) -- (12,0) node[midway,above]{$\ldots$};

			% File o_1	
			\drawfile{19}{20}{$\file(o_1)$}{$\frac{w_{o_1}}{2}$};
			
			% File z_1	
			\drawfile{16}{19}{$z_1$}{$k$};
			
			% File o_2	
			\drawfile{15}{16}{$\file(o_2)$}{$\frac{w_{o_2}}{2}$};
			
			% File z_2	
			\drawfile{12}{15}{$z_2$}{$k$};
			
			% W
			\draw[<->] (5+0.025,-1.25) -- (20-0.025,-1.25) node[midway,below]{$W = k(|O|-1) + \sum\limits_{o \in O}\frac{w_o}{2}$};

			\end{tikzpicture}
			\caption{Reduction of Knapsack Problem to~\LTSPR used in the proof of Theorem~\ref{thm:ltsprhard}.}
			\label{fig:knapsacktape}
		\end{figure}

		Let~$\minibatches_1 \subseteq \minibatches$ be the set of (not necessarily atomic) mini-batches selected for execution in Phase 1, $\Delta_t(\minibatches_1) = \sum_{\minibatch \in \minibatches_1}2\size(\minibatch)$  be the number of time-steps by which the duration of Phase 1 increases due to their execution, and $V_1 =  \sum_{\minibatch \in \minibatches_1}\nrequests(\minibatch) $ be the number of requests serviced in Phase 1 according to~$\minibatches_1$. 
		%Moreover, let~$\files_1$ denote the set of files read in Phase 1 due to the execution of all mini-batches in~$\minibatches_1$. 
		Note that  $\sum_{\file \in \files(\minibatches_1) } 2s(\file) \leq \Delta_t(\minibatches_1)$, i.e., each file read in Phase 1 is traversed at least twice.  
		
		% Let~$\minibatches_1^*$, $\files_1$, $k_1^*$, and~$V_1^*$ denote sets and values of an optimal solution for~$I'$.

		Let $V = \sum_{o \in O}v_o$ denote the sum of the values of all elements in~$O$. Assume that $\nrequests(d) = 3V(W + \size(d) + k)$ and that all requests in~$\setrequests(d)$ are released at time step  $W + \size(d) + k$. We claim that, for suitable values of~$\size(d)$, the set~$\files(\minibatches_1)$ selected by any exact algorithm solving~\LTSPR is optimal for~$I'$ if and only if~$O(\files(\minibatches_1))$ is optimal for~$I$.

		First, observe that if $\minibatches_1 = \emptyset$ and Phase 2 starts at time $W + \size(d)+k$, the resulting solution has an overall completion time upper-bounded by~$2V(W + \size(d) + k)$, i.e., requests in~$\setrequests(d)$ are executed immediately after their release whereas all the other~$V$ requests (those associated with knapsack items) will wait at most $W + \size(d) + k$ additional time-units after the beginning of Phase 2. Note that this strategy is feasible, as the tape needs only $W + \size(d)$ time-steps to move to position 1. Conversely, if $\minibatches_1 \neq \emptyset$ and $\Delta_t(\minibatches_1) > k$, Phase 2 will not start before time $W + \size(d) + k + 1$, so the the overall completion time is at least~$3V(W + \size(d) + k)$, as a result from the (strictly positive) waiting time of requests in $\setrequests(d)$. Therefore,  by construction, $\Delta_t(\minibatches_1) \leq k$ for any optimal solution of~$I'$. Moreover, an optimal~$\minibatches_1$ consists solely of atomic mini-batches, as the distance between two knapsack files equals~$k$.
		%, and each file in~$\files_1$ is read exactly once. 
		
		We identify now conditions that should be fulfilled by~$\size(d)$ in order to enforce the inclusion of as many requests associated with knapsack files 
		%in $\files \setminus \{d\}$ 
		as possible in an optimal solution of~$I'$. We claim that this property holds  if $\Delta_t(\minibatches_1)V < 2\size(d)$. %for $\Delta_t(\minibatches_1) \leq k$. 
		The expression on the left-hand side is an upper bound on  the maximum increase in waiting times 
		%for requests associated with knapsack files; %in $\files \setminus \{d\}$; 
		(recall that elements of~$\setrequests(d)$ should not be included in this analysis, as they will only be released afterwards, and that artificial files have no associated request). The term  on the right-hand side is a tight lower bound on the reduction of the waiting time for a single request serviced in Phase 1 (in comparison with that delivered by scheduling plan in which  $\minibatches_1 = \emptyset$).  Thus, if $\size(d) > \frac{kV}{2}$, the execution of any read request in Phase 1 will reduce the overall completion time.
		
		By construction, an optimal solution for~$I'$ consists of a set~$\minibatches_1^*$ such that $V_1^* = \sum_{\file \in \files(\minibatches_1^*)}\nrequests(\file)$ is maximized whereas $\Delta_t(\minibatches_1^*) =  2\sum_{\file \in \files(\minibatches_1^*) }\size(\file) \leq k$.  Moreover, 
		set~$O(\files(\minibatches_1^*))$ is such that $\sum_{o \in O(\files(\minibatches_1^*))}w_o \leq k$ and $\sum_{ o \in O(\files(\minibatches_1^*)) }v_o = V_1^*$, so solutions selected by an exact algorithm of~\LTSPR with release times are associated with feasible solutions for~$I$. %whose values are given by~$V_1$. 
		As every solution for~$I$ can also be directly reduced to a solution for~$I'$, it follows that both problems are equivalent. Finally, similar arguments can be used to show the equivalence between the respective decision versions of both problems, so we conclude that~\LTSPR  is NP-complete. %$\qed$
	\end{proof}
	%\begin{corollary}
	%	The offline \LTSP with release times restricted to atomic mini-batches is NP-complete.
	%\end{corollary}

	%\subsection{Online Linear Tape Scheduling Problem}
	
	We investigate now~\OLTSP, the online extension  of~\LTSPR. The following results shows that there is no algorithm capable of delivering solutions differing from the (offline) optimal by a constant multiplicative factor. 
	
	\begin{proposition}[\cite{cardonha2016online}, Proposition 5]\label{prop:onlinebad}
		There is no $c$-competitive algorithm for~\OLTSP for any constant~$c$.
	\end{proposition}
	\begin{proof}
		Let~$A$ be an arbitrary algorithm for~\OLTSP, and let us consider a scenario in which~$\files = (\file_1,\file_2)$, with $\fileleft(\file_1) = 1$, $\size(\file_1) = k$, $\fileleft(\file_2) = k+1$, and $\size(\file_2) = k^2$.
		As $\lengthtape = k^2 + k$, any position of the tape can reached within~$O(k^2)$ time-steps.   Let us suppose that some request~$\request$ in $\setrequests(\file_2)$ is released. 
		
		If the time spent by~$A$ to start servicing~$\request$ is~$\omega(k^2)$, no more requests are released. For this scenario, an optimal algorithm would start servicing~$\request$ in time~$O(k^2)$,
		%, regardless of the original position of the tape, 
		thus yielding a solution whose ratio with that delivered by~$A$ is not bounded by any constant value.
		
		Conversely, if $A$ starts traversing~$\file_2$ in time~$O(k^2)$, $k$ requests associated with~$\file_1$ are released as soon as the tape reaches position~$\fileright(\file_2)$ after executing~$\file_2$. In this case, the overall response time produced by~$A$ would be $O(O(k^2) + k(k^2 + k)) = O(k^3)$, whereas an optimal solution would have moved the tape to~$\fileleft(\file_1)$ first  and then  towards~$\fileright(\file_2)$ immediately after the release of requests in~$\setrequests(\file_1)$, with an overall response time of $O(k^2) + k = O(k^2)$. In this case, the ratio between the solution of~$A$ and the optimal solution also cannot be bounded by any constant value. 
		
		Since no assumption was made about~$A$, we conclude that there is no algorithm able to produce $c$-competitive solutions for~\OLTSP for any constant~$c$. %$\qed$ 
	\end{proof}
	
	\begin{comment}
	The \textsc{Local Selection Strategy (\textbf{LSS})}, presented in Algorithm~\ref{alg:localselectionstrategy}, can be interpreted as a local decision extension of Algorithm~\ref{alg:offlineselectivegreedy} in which the tape head is positioned in an arbitrary position~$h$ of the tape and a decision is made about the execution of a mini-batch starting from~$h$. \textbf{LSS} will be useful in the online~\LTSP.
	
	\alglanguage{pseudocode}
	\begin{algorithm}[!ht]
	\caption{\textsc{Local Selection Strategy (\textbf{LSS})}}
	\label{alg:localselectionstrategy}
	\begin{algorithmic}
	\State \file = \file(h)
	\State \minibatch = (\file,\file')
	\For{$\file' \in \files$ such that $\fileleft(\file) \leq \fileleft(\file')$}
	%		\LineComment{\emph{Check condition of Proposition~\ref{prop:executeinphase1}}}
	\If{ 
	$\nrequests(\file)\left(\fileleft(\file) +  \sum_{ \substack{\minibatch \in \minibatches_1\\ \minibatch < \file }}\size(\minibatch)  \right) 
	>
	\size( \minibatch  )\left(  \sum\limits_{ \substack{\file'' \in \files\\ \file'' \notin \minibatch }}\nrequests(\file'')        \right)$ 
	}
	\State $\minibatches_1 = \minibatches_1 \cup \{(\file,\file)\}$
	\EndIf
	\EndFor	
	\end{algorithmic}
	\end{algorithm}
	
	\end{comment}
	
	%\section{Algorithms for the online LTSP}\label{sec:online-algorithms}
	
	The algorithms for \OLTSP presented in this work are \textit{adaptive},  since incoming requests are immediately incorporated and considered for scheduling; that is, we assume that~$\setrequests$ is updated every time the leftmost bit of a file is reached. 
	%Note that frequent updates may lead to job starvation; for instance, a sequence of incoming jobs with a significant number of requests associated with files positioned in the end of the tape could make the tape to continuously move to the right, thus leading jobs located in the beginning of the tape unserviced. However, since tape storage is usually at the bottom layer of a hierarchical storage management system (that is, data retrieved from tape is cached on a higher speed device such as a hard disk drive), starvation is not likely to happen in practice; all possible blocks of data positioned to the right of~$h$ would eventually hit the cache and the remaining requests would be served. Moreover, tapes are typically used in scenarios where read requests are not so frequent, thus making such pathological behavior unlikely in practice.
	In order to represent relevant dynamic features of~\OLTSP, we introduce additional notation. The current position of the tape is denoted by~$h$, and~$\file(h)$ denotes the file whose leftmost bit is~$h$; we also say that $h$ is over file~$\file(h)$. Function~$\file(h)$ is not defined in other positions, but this aspect is irrelevant in this work, as preemption is not considered and decisions are not made when the tape is positioned in ``intermediate'' blocks. Two operations may take place when~$h$ is over file~$\file$: $\mathbf{ReadFile}(\file)$ and $\mathbf{ExecuteMiniBatch}((\file,\file'))$ for some file~$\file' \in \files$ such that $\file' \geq \file$. $\mathbf{ReadFile}(\file)$  denotes the rightward traversal of~$\file$; as a consequence of this operation, all requests currently in $\setrequests(\file)$ are serviced and~$h$ changes from $\fileleft(\file)$ to~$\fileright(\file)$. $\mathbf{ExecuteMiniBatch}((\file,\file'))$ represents the execution of mini-batch $(\file,\file')$, i.e., the tape moves from~$\fileleft(\file)$ to~$\fileright(\file')$ and back, so~$h$ will be on~$\fileleft(\file)$ by the end of its execution. 
	\subsubsection*{Linear Tape File System Strategy (\LTFS)}
	Algorithm~\ref{alg:uniquequeue} contains the pseudo-code for the current strategy widely used in the industry by implementations of the Linear Tape File System and, in particular, in the real-world scenarios we investigate in our computational experiments~(\cite{ISO20919}).  As input parameters, \LTFS receives the list~$\setrequests$ of pending requests and the current position~$h$ of the tape. The algorithm employs a FIFO (first-in, first-out) strategy based on the release times of requests in~$\setrequests$ to define its movements, that is, it moves to the left or to the right depending on the current position~$h$ and on the leftmost bit of the file containing the oldest pending request;
	if~$\setrequests = \emptyset$,  \LTFS stops and stays still until a new request is released. \LTFS acts in an opportunistic way, in the sense that all pending requests in~$\setrequests(\file(\request))$ are serviced together with~$\request$. On the other hand, whenever the tape moves to the right, requests associated with files traversed during the movement are ignored and not executed, regardless of their release times; the scheduler only looks into serving the currently given request. Finally, if several requests are released simultaneously, \LTFS applies a random service order to them.
	
	\alglanguage{pseudocode}
	\begin{algorithm}[!ht]
		\caption{Linear Tape File System Strategy (\LTFS)}
		\label{alg:uniquequeue}
		\begin{algorithmic}
			\Procedure{$\LTFS$}{$\mathcal{J},h$} 
			%\LineComment{\emph{Infinite loop for online scenarios}}
			\While{$\mathcal{J} \neq \varnothing$ }
			\LineComment{\emph{Select request with minimum release time}}
			\State $\request := \min\limits_{\request \in \setrequests}\releasetime(\request)$
			%        \LineComment{\emph{Select first request in $\mathcal{J}$}}
			%        \State $\request := \mathcal{J}_0$
			\LineComment{\emph{Tape moves to the leftmost bit of $\file(\request)$ }}
			\State $h := l(f(j))$
			\LineComment{\emph{Execute file}}
			\State $\mathbf{ReadFile}(\file(h))$
			\EndWhile
			\EndProcedure
		\end{algorithmic}
	\end{algorithm}

	As an improvement to the original \LTFS algorithm, we propose the Enhanced Linear Tape File System Strategy (\ELTFS), which processes pending requests of files traversed while the tape moves rightwards from position $h$ to the leftmost bit of the target file; by construction, the objective value of the scheduling plan produced by \ELTFS is never larger than the one yielded by~\LTFS. %LTFS+ is presented in Algorithm~\ref{alg:uniquequeue_plus}.
	
	%\alglanguage{pseudocode}
	%\begin{algorithm}[!ht]
	%	\caption{Enhanced Linear Tape File System Strategy (LTFS+)}
	%	\label{alg:uniquequeue_plus}
	%	\begin{algorithmic}
	%		\Procedure{$\mathbf{LTFS+}$}{$\mathcal{J},h$} 
	%		\While{$\mathcal{J} \neq \varnothing$ }
	%		\LineComment{\emph{Select request with minimum release time}}
	%		\State $\request := \min\limits_{\request \in \setrequests}\releasetime(\request)$
	%		\LineComment{\emph{Tape moves to the leftmost bit of $\file(\request)$ }}
	%		\If {$h > \fileleft(\request)$}
	%			\State $h := l(f(j))$
	%		\EndIf
	%		\LineComment{\emph{As tape moves to the right, pending requests are serviced}}
	%		\While {$h < \fileright(\request)$}
	%			\State $\mathbf{ReadFile}(\file(h))$ %$h := \fileright(\file(h)) + 1$
	%		\EndWhile
	%		\EndWhile
	%		\EndProcedure
	%	\end{algorithmic}
	%\end{algorithm}
	
	% % % % % % % % % % % % % % % % %
	% ONLINE FILTERED  GREEDY   % % %
	% % % % % % % % % % % % % % % % %
	\subsubsection*{\textbf{Replan}}
	In Algorithm~\ref{alg:replan} we describe \Replan, a novel procedure for~\OLTSP. Function 
	$\LTSP(h)$ returns a schedule $\minibatches_1 \subseteq \minibatches$ for the instance~$I$ of~\LTSP  defined by the set of pending requests and the current position~$h$ of the tape, i.e., \Replan relies on an algorithm for~\LTSP. 
	%We can assume that~$\minibatches_1$ does not contain mini-batches~$\minibatch$ such that $\fileleft(\minibatch) > h$, as the subset $\minibatches_1' = \{\minibatch \in \minibatches_1: \fileleft(\minibatch) \geq h\}$  could be substituted for the mini-batch $(h,\max_{\minibatch \in \minibatches_1'}\fileright(\minibatch))$.
	
	\alglanguage{pseudocode}
	\begin{algorithm}[!ht]
		\caption{\Replan}
		\label{alg:replan}
		\begin{algorithmic}
			\Procedure{$\Replan$}{$\mathcal{J},h$} 
			%		\While{$\mathcal{J} \neq \varnothing$ }
			%		\LineComment{\emph{Select request with minimum release time}}
			%		\State $\minibatch^* :=  \min\limits_{  \substack{ \minibatch  \in \minibatches:  \fileleft(\minibatch) = h  }   } v(\LTSP(\setrequests, h, \{\minibatch\}))$
			%		\If {$ v(\LTSP(\setrequests, h, \{\minibatch^*\})) <  v(\LTSP(\setrequests, h, \emptyset ))  $}
			%		\State $\mathbf{ReadFile}(\file(h),h)$
			%		\Else
			%		\State $h := \fileleft( \file(h)-1   )$
			%		\EndIf
			%		\EndWhile
			%		\EndProcedure
			\LineComment{\emph{Phase 1: move towards leftmost file with pending requests}}
			\While{$\file(h) \neq \min\limits_{(\file,\file') \in \minibatches_1}\file$ }
			\State $\minibatches_1 :=  \LTSP(h)$
			\If {$ \exists \file' \in \files$ s.t. $(\file(h),\file') \in \minibatches_1  $}
			\State $\mathbf{ExecuteMiniBatch}((\file(h),\file'))$
			\EndIf
			\State $h := \fileleft(\file(h)-1)$
			\EndWhile
			
			\LineComment{\emph{Phase 2: move towards rightmost file with pending requests}}
			\While{$\sum\limits_{\file \geq \file(h)}\nrequests(\file) > 0$ }
			\State $\mathbf{ReadFile}( \file(h) )$
			\EndWhile
			\EndProcedure
		\end{algorithmic}
	\end{algorithm}
	
	Following the structure of optimal solutions for~\LTSP, \Replan iteratively executes a two-phases procedure in order to service requests. In Phase 1, the tape moves towards the first block of the leftmost file containing at least one pending request. For this phase, the algorithm maintains a schedule~$\minibatches_1$, which may be updated in each step (or after the arrival of each new request) through the application of some algorithm that solves~\LTSP. Set~$\minibatches_1$ contains only mini-batches starting at the current position or to its left, for reasons we discuss below.
	%is used to orchestrate the execution of mini-batches, so movements to the right (and back) may take place during Phase 1.  
	After finishing Phase 1, \Replan executes Phase 2, in which the tape moves to the right until the rightmost file with pending requests. Afterwards, the algorithm returns to Phase 1 and repeats the process until no more pending requests are left.

	Preliminary experiments involving early versions of \Replan showed that one should be careful %when designing algorithms for~\OLTSP 
	in order to avoid the so-called ``shoe-shining'' effect, which takes place when the tape spends a long time moving back and forward over a restricted region (typically close to its center). Shoe-shining may lead to long response times (i.e., starvation) of requests associated with files located in the beginning and in the end of the tape, so it should be avoided. 
	%In particular, response times for requests associated with files located in the extremities of the tape may become very large if Phase 1 or Phase 2 are constantly interrupted, a situation whose negative effects may be magnified if files receive many requests distributed over long periods of time (i.e., in cases where the tape needs to move from the beginning to the end and back several times). 
	Therefore, \Replan does not execute more than one mini-batch starting from any given file in each iteration of Phase 1, i.e., after executing some mini-batch~$\minibatch$, Replan always proceeds to position~$\fileleft(\minibatch)-1$.

	% that decides in each step the direction of the next movement based on the best schedule~$\minibatches_1$ identified for the current instance~$I$. In order to find~$\minibatches_1$, Replan computes $v(\LTSP(\setrequests, h, \{\minibatch\})$ for each $\minibatch$ in $\minibatches$ starting at position~$h$ and selects the mini-batch~$\minibatch^*$ delivering the best result. If $ v(\text{N-FGS}(\setrequests,h, \{\minibatch^*\})) <  v(\text{N-FGS}(\setrequests,h,\emptyset ))$, $\minibatch^*$ belongs to~$\minibatches_1$ and file~$\file(h)$ is executed immediately, i.e., the tape moves to the right. Otherwise, the tape moves to the left.

	%identified by any strategy for the instance~$I$ of~\LTSP defined by the current status of the tape, given by the set of pending requests~$\setrequests$ and the tape position~$h$. Namely, if $\minibatches_1$ contains a mini-batch~$\minibatch$ such that $\fileleft(\minibatch) = \file(h)$, the head moves to the right; otherwise, it moves to the left. In order to find~$\minibatches_1$, Replan inspects the value of the schedule delivered by any strategy of \LTSP for~$I$ using $\{\minibatch\}$ as the set of fixed mini-batches for each $\minibatch$ in $\minibatches$ starting at position~$h$; if $ v(\text{N-FGS}(\setrequests,h, \{\minibatch\})) <  v(\text{N-FGS}(\setrequests,h,\emptyset ))$, $\minibatch$ improves the quality of the solution and file~$\file(h)$ is executed immediately.

	In our computational experiments, we evaluate the combination of \Replan with some of  the algorithms we presented for~\LTSP. Namely, we assess the performance of~\ReplanSSS, \ReplanGS, \ReplanFGS, and \ReplanLogNFGS, which denote the combination of \Replan with \SSS, \GS, \FGS, and \LogNFGS, respectively. Below, we briefly describe  each of these algorithms.

	\paragraph{\textbf{Replan+SSS}:} This algorithm maintains and executes a single mini-batch $(\file,\file')$, where $\file$ and $\file'$ are the leftmost and rightmost files, respective, associated with  at least one pending request. Phase~1 and Phase~2 finish when~$\fileleft(\file)$ and~$\fileright(\file')$ are reached, respectively. Both Phase~1 and Phase~2 can only be extended during their execution, and not interrupted; for example, if a request associated with file~$\file'' < \file$ is released before the end of Phase~1, the new target destination becomes~$\fileleft(\file'')$.

	\paragraph{\textbf{Replan+GS}:} The difference between \ReplanGS and \ReplanSSS is that the former executes an atomic mini-batch in Phase 1 for every file containing  pending requests. Phase 2 is analogous.
	%After executing a mini-batch to execute file~$\file$, Replan necessarily moves to file $\file-1$, thus avoiding starvation in worst-case scenarios (e.g., an adversary may keep releasing new requests for~$\file$ in order to keep the tape stalled at its position).
	
	\paragraph{\textbf{Replan+FGS}:} In each step (of after the arrival of a new request), \ReplanFGS updates its schedule using \FGS. Therefore, if file~$\file(h)$ satisfies Corollary~\ref{prop:executeinphase2}, its execution is postponed to Phase~2. 
	%(either to Phase 2 or until some mini-batch~$(\file,\file')$ such that $\file < \file(h) \leq \file'$  is executed). 

	\paragraph{\textbf{Replan+LogNFGS}:} \NFGS is computationally expensive,  so we employ solely the combination of \Replan with \LogNFGS, which is asymptotically faster. In each iteration, \ReplanLogNFGS maintains the schedule obtained by \FGS and, by the time file~$\file$ is reached in Phase 1, the opportunity cost of each possible mini-batch starting from~$\file$ and containing at most~$\log(\files)$ files is evaluated. If the smallest opportunity cost is negative, the associated mini-batch is executed; otherwise, the execution of~$\file$ is postponed.

	\section{Computational experiments}\label{sec:experiments}
	
	%arquivos de mesmo tamanho - dah para ter varios workers
	%arquivos de tamanho variavel e no max um request - sistemas de arquivo (setar prioridade com num requests para arquivos)

	We present in this section the results of the computational evaluation we performed on the algorithms described for \LTSP and \OLTSP. All experiments were executed  on an Intel(R) Xeon(R) CPU E5-2680 v2 at 2.80GHz. We used the Wilcoxon signed-rank test in order to estimate $p$-values comparing pairs of algorithms (\cite{wilcoxon1945individual}). %; the sequence of values used in a test contains the sum of the response times for each instance used in the experiments.
	All the algorithms were implemented in C++. Overall response times are always reported as relative values; namely, these entries contain the sum of the response times for all instances of the associated configuration (identified in the first column of the same row)  yielded by the respective algorithm  (identified in the first line of the same column) divided by the sum of the response times obtained by a baseline algorithm, defined per class of instances.
	%\LTFS (or \ELTFS in the case of the Synthetic benchmark) for the same (set of) instance(s). 
	The lower the value, the better, and the best result is marked in bold. Significant running times are reported in seconds, and we did not incorporate the time spent with these procedures into response times; in particular, for the online scenarios, objective values were evaluated as if all scheduling sub-procedures would take a negligible amount of time. 
	%An explanation for this is the fact that certain computational tricks may be used in the implementation of these schedules in order to mitigate the effects of this scheduling times (in scenarios where these scheduling times are sufficiently small), so we decided not t

	%We compare LTFS' performance with some of the algorithms proposed in this paper by replaying the produced traces on a simulator.

	%The entries of the leftmost column of all tables presented in this section contain a description of the associated family of instances using the format $|\files|$-$|\timehorizon|$, with~$|\files|$ indicating the \textit{number of files}~$|\files|$ and~$|\timehorizon|$ the length of the \textit{time-horizon}  within which requests were released. The other columns indicate the algorithms being tested. 

	\subsection{Instances}
	
	For our experiments, we used two  sets of instances: \textit{Landsat} and \textit{Synthetic}. The Landsat benchmark was obtained from a real-world application involving remote sensing for precision agriculture, whereas
	%The dataset extracted  from a real-world scenarios was extracted from traces and log files generated by LTFS. 
	the Synthetic benchmark was generated to stress the algorithms. The same instances were used in the experiments of~\LTSP and~\OLTSP; for~\LTSP, release times are ignored (or, similarly, assumed to be equal to zero). Details about each class of instances are presented below.

	\subsubsection*{Landsat instances}
	
	%For the real-life offline workloads we handle a use case centered around the use of remote sensing for precision agriculture in 3 zones of interest. The tape archive covers two years and a half worth of satellite data. The satellite scenes used in the workloads are represented by files holding multi-spectral data obtained from the Landsat 8 space mission, which images the entire Earth every 16 days. Because the entire scanning process generates large amounts of data, Landsat datasets are offered as a collection of \emph{tiles}. Each tile features 11 bands representing different ranges of the electromagnetic spectrum and occupy an approximate size of 1,7 GB of storage~(\cite{Roy14}).
	
	%sThe Landsat instances have been obtained from real-world scenarios involving remote sensing for precision agriculture. These instances 
	Landsat instances consist of  multi-spectral data obtained from the Landsat 8 space mission, which images the entire Earth every 16 days. Because the scanning process generates large amounts of data, Landsat datasets are offered as collections of \emph{tiles} (or satellite scenes). Each tile is stored in a single file with approximately  1.7 GB of data and features 12 \emph{bands}, which are matrices of  pixels; 11 bands represent different ranges of the electromagnetic spectrum, and one additional band (particular to Landsat), called QA (for Quality Assessment), classifies pixels, with an accuracy of up to 80\%, as water, snow, ice, cloud, or as invalid due to sensor errors during the image acquisition process~(\cite{Roy14}).

	We used data from 3 zones of interest in our experiments, corresponding to vineyards in different parts of the world: the Atacama desert in Chile, the Serra Gaucha wine region in the south of Brazil, and the Manduria region in Italy. The Atacama region, in special, is known for its extremely dry weather and a mostly cloud-free sky all year round, which translates into crisp satellite imagery. Conversely, the Serra Gaucha and Manduria regions are frequently covered with clouds that, at times, invalidate whole satellite scenes~(\cite{CloudyEarth}).

	The Atacama region  is represented by a cluster of 4 Landsat tiles. The software accessing these tiles computes the Normalized Difference Vegetation Index (NDVI) of the region by reading Landsat bands number 4 and 5, which represent the reflective radiation in visible Red and Near-Infrared wavelengths, respectively. The Serra Gaucha region is covered by a collection of 6 Landsat tiles that, once processed, are converted into false-color images that emphasize vegetation; the image composition procedure used in this operation  requires the use of bands number 4, 5, and 6. The last zone of interest, Manduria, is represented by a set of 5 tiles; the application reading them feeds all 11 bands into a deep neural network that helps classify the vineyards'
	health conditions.
	
	All software processing the aforementioned zones look into the QA band before launching parallel readers for the actual bands of interest. If a scene has a cloud cover above a certain threshold, then no bands other than the QA are read; the missing data is then interpolated or extrapolated at a later time using data from past and future scenes. An implementation detail worth noting is that parallel readers are orchestrated by a high-performance I/O interface implemented on top of Message Passing Interface (MPI)~(\cite{Thakur99}). As a consequence of the collective I/O techniques employed by MPI, release times of jobs related to bands within a same satellite scene are always identical.

	All instances of the Landsat benchmark contain 873 files. The in-house software used in this study always reads bands in parallel through the use of concurrent operating system threads; consequently, we may have up to 12 read requests associated with the same file being triggered simultaneously (one for each band).
	%One individual job is released for each band requested, so the same file may receive two or more read requests from the same software service. 
	For example, the NDVI software issues 2 requests (for bands 4 and 5) for each tile composing the Atacama region. The volume of read requests is influenced by the occurrence of clouds in each zone of interest; namely, we can say that scenarios with less clouds will have more read requests. Therefore, each instance of the dataset is defined by a triple $\alpha$-$\beta$-$\gamma$; numbers $\alpha$, $\beta$, and~$\gamma$ belong to the interval~$[0,100]$ and indicate the threshold for cloud covers in Atacama, Serra Gaucha, and Manduria, respectively, on that particular scenario.

	The scenarios covered by the Landsat benchmark employ 10 combinations of $\alpha$-$\beta$-$\gamma$;  there are 10 instances per combination, so this dataset is composed of 100 instances. Finally, \LTFS is used as the baseline algorithm in the experiments involving for this benchmark; as its results  were two orders of magnitude higher than those produced by the other algorithms, values reported in the tables equal the actual ratios  multiplied by 1000.

	\subsubsection*{Synthetic instances}
	
	Each instance of this dataset is parameterized by a pair~$|\files|$-$k$, $k \in \{1,3,5\}$, where~$k\lengthtape$ is the length of the time horizon~$\timehorizon$ within which requests were released. The size~$\size(\file)$ of each file~$\file$  is drawn uniformly from~$[1,20]$; from these values, we obtain~$m$ and, together with~$k$, also~$\timehorizon$. In order to generate sequence~$\setrequests(\file)$, we initially drawn an integer value~$\lambda_f$ uniformly from interval $[|\timehorizon|/50,|\timehorizon|/5]$; values for different files in~$\files$ are independent and identically distributed. Afterwards, we drawn~$|\timehorizon|/\lambda_f$ samples from the Poisson distribution parametrized with~$\lambda_f$ and generate the sequence of release times for requests in $\setrequests(\file)$ by setting the arrival time of the $i$-th request as the sum of the first~$i$ values drawn from the distribution; requests whose release times overpass~$|\timehorizon|$ are discarded.  Finally, to reduce bias and make the results of our experiments  statistically significant, we generated 10  instances for each configuration $|\files|$-$k$. Note that, by construction, the number of requests per file is not strongly affected by~$k$; namely, time intervals between releases are expected to increase proportionally to~$k$, but the number of requests released over the whole time horizon are likely to be similar.
	
	Typically, for $k = 1$,  the algorithms need to traverse the whole tape only once. In order
	%results of the Replan algorithms are very similar to those of their  offline counterparts, as release times are concentrated within a relatively short time horizon and the online algorithms gain access to a complete description of the instance before major bad decisions are taken. Therefore, in order 
	to properly stress the online algorithms, we also generated instances for larger values of~$k$. We remark that scenarios where~$k$ is very large are not expected  in practice, though, as magnetic tapes are typically not the technology of choice in scenarios where the tape needs to be completely traversed several times within a relatively short period of time (i.e., within minutes), so we restrict the benchmark to configuration with $k \leq 5$.
	
	\LTFS delivered very bad results for all  instances of this dataset, both in the offline and in the online cases, losing for all the other algorithms by 4 orders of magnitude. Therefore, in order to make the differences in performance between the other algorithms more explicit, we decided to report the relative results for Synthetic instances with respect to \ELTFS instead of \LTFS. Moreover,  the reported relative results equal the actual ratios multiplied by~100.

	\subsection{Offline algorithms}

	We compare the performance of \LTFS with all algorithms for \LTSP that have been presented in this work:  \SSS (Simple Schedule Strategy), \GS (Greedy Strategy), \FGS (Filtered Greedy Strategy), \NFGS (Non-atomic Filtered Greedy Strategy), \LogNFGS, and \ELTFS (Enhanced Linear Tape File System Strategy). 
	%The original datasets were generated for the online scenario, so requests have non-trivial release times. In order to use them as input for these algorithms, we ignore release times and assume that all requests are made available for execution in the beginning of the planning horizon, i.e., we assume that $\releasetime(\request) = 0$ for every request~$\request$. 
	%All requests are known and can be serviced right from time-step zero by all algorithms; response times are also computed assuming that all jobs were released at time-step zero. 
	In order to avoid a completely random behavior of \LTFS and \ELTFS, the actual release times (which are ignored for the tests involving \LTSP) are used by both algorithms to define the order with which requests should be executed.

	%Tables~\ref{tab:offsynthetic}, \ref{tab:offanti-virus}, and~\ref{tab:offagriculture} contains the results for the synthetic, Anti-virus, and Landsat instances, respectively. 

	%\begin{itemize}
	%	\item relative performance between algorithms
	%	\item maximum size tractable within 1 second
	%	\item impact of time horizon
	%	\item impact of number of files
	%\end{itemize}

	\subsubsection*{Synthetic instances}

	The performance of the offline algorithms on the instances of the  Synthetic dataset are presented in Table~\ref{tab:offsynthetic}. In addition to the relative performance of the proposed algorithms with respect to \ELTFS, we also report the average running times of \NFGS and \LogNFGS; all the other strategies spent less than 0.01 seconds to solve each instance of the dataset. 
	%In average, the solutions produced by LTFS were worse by 4 orders of magnitude, so we multiplied the entries of Table~\ref{tab:offsynthetic} by $10^4$ in order to facilitate a more direct comparison between the proposed algorithms. In summary, all the algorithms proposed in this article clearly over-perform  LTFS. %; namely, there was only one instance for which LTFS delivered a better solution than GS.
	\NFGS gave the best results, but both \LogNFGS and \FGS had very similar performances; we conclude that \FGS seems to be more adequate for large-scale practical scenarios of~\LTSP, as its running time is negligible and the quality of its solutions is almost as good as those of \NFGS. 
	% optimized the trade-off between running time and solution quality and seems to be more adequate for practical scenarios, as discussed below. 

	\begin{table}
		\footnotesize
		\centering
		\caption{Performance of algorithms for \LTSP on Synthetic instances.}
		\label{tab:offsynthetic}
		\begin{tabular}{|c|ccccc|cc|}
			\hline
			& \multicolumn{5}{c}{Relative objetive value ($\times 10^2$)  }  & \multicolumn{2}{c}{Time (sec.)} \\
			Instance   &    \SSS &     \GS &    \FGS & \NFGS          &   \LogNFGS &    \NFGS &   \LogNFGS \\
			\hline
			20000-1    & 80.53 & 80.71 & 64.82 & \textbf{64.82} &      64.82 &  1.634 &      0.641 \\
			20000-3    & 89.71 & 89.94 & 72.20  & \textbf{72.20} &      72.20  &  1.628 &      0.643 \\
			20000-5    & 83.22 & 83.35 & 66.92 & \textbf{66.92} &      66.92 &  1.627 &      0.636 \\
			40000-1    & 83.30  & 83.19 & 66.87 & \textbf{66.87} &      66.87 &  6.893 &      2.874 \\
			40000-3    & 82.95 & 83.05 & 66.71 & \textbf{66.71} &      66.71 &  6.889 &      2.852 \\
			40000-5    & 87.96 & 87.71 & 70.61 & \textbf{70.61} &      70.61 &  6.943 &      2.899 \\
			60000-1    & 82.24 & 82.33 & 66.15 & \textbf{66.15} &      66.15 & 15.855 &      6.720  \\
			60000-3    & 91.99 & 92.07 & 73.93 & \textbf{73.93} &      73.93 & 15.948 &      6.758 \\
			60000-5    & 86.10  & 86.15 & 69.24 & \textbf{69.24} &      69.24 & 15.619 &      6.623 \\
			80000-1    & 79.40  & 79.42 & 63.82 & \textbf{63.82} &      63.82 & 28.253 &     12.045 \\
			80000-3    & 79.27 & 79.23 & 63.68 & \textbf{63.68} &      63.68 & 28.582 &     12.195 \\
			80000-5    & 83.53 & 83.56 & 67.11 & \textbf{67.11} &      67.11 & 28.598 &     12.213 \\
			100000-1   & 79.50  & 79.53 & 63.91 & \textbf{63.91} &      63.91 & 44.227 &     18.841 \\
			100000-3   & 83.27 & 83.17 & 66.87 & \textbf{66.87} &      66.87 & 42.744 &     18.304 \\
			100000-5   & 88.37 & 88.33 & 71.00    & \textbf{71.00} &      71.00    & 43.986 &     18.887 \\
			\hline
		\end{tabular}
	\end{table}

	From the theoretical standpoint, we were not able to prove a full dominance relationship involving \SSS and the remaining algorithms; in particular, we presented a family of instances in the proof of Proposition~\ref{prop:greedy3approx} for which \SSS delivers better results than \GS. Our statistical analysis shows that the performance of \GS was not significantly different from that of \SSS  ($p$-values equal to 0.93). Moreover, \SSS  delivered better  results in 74 scenarios whereas \GS was the best in 76, thus suggesting a strong similarity of performance between both algorithms. As  \GS is a 3-approximation for \LTSP (Proposition~\ref{prop:greedy3approx}), the results suggest that the pathological scenarios that hindered the derivation of such performance guarantees for \SSS (see Proposition~\ref{prop:simplebad}) are not so likely to emerge  by chance.

	By construction, \NFGS dominates \FGS and \FGS dominates \GS, so we verified these relationships from the practical standpoint. First, we saw a clear dominance of \FGS over \GS. The $p$-value obtained via the Wilcoxon signed-rank test was $10^{-25}$, and \FGS delivered better results than \GS in all instances; moreover, in average, overall response times of \FGS were 19.6\% better than those of \GS. The differences between \FGS and \NFGS were also statistically significant ($p$-value of $10^{-12}$), but the improvements were not so remarkable. Out of the 150 individual instances, \NFGS improved the schedule over \FGS   in  82 cases; in average, though, improvements were inferior to $0.0005\%$, and in the cases were they did occur, the number grows to only~$0.0009\%$. This behavior can be  observed in Table~\ref{tab:offsynthetic}, as the relative performances of \FGS and \NFGS were similar for a precision of 4 digits.  Moreover, all strategies employing filtering (i.e., Corollary~\ref{prop:executeinphase2}) clearly dominated \SSS and \GS, with $p$-values of $10^{-25}$ and average improvements of almost $20\%$.

	The running times of \NFGS and \LogNFGS show that these two strategies may not be suitable for practical applications if the number of files is large.
	%In terms of running time, as expected, both NFGS and Log-NFGS are severely impacted by the number of files in the tape. 
	For instance, scenarios with 100,000 files need at least 15 seconds to be solved by these algorithms (more than 40 seconds in the case of \NFGS), which is prohibitive in certain real-world environments. The improvements of these two algorithms over~\LTFS are significant and could compensate for these delays, but~\FGS delivers schedules which are almost as good (frequently equal) within 0.01 seconds, which makes it a better choice as a general-purpose algorithm for~\LTSP. 
	%Given that the reductions in the overall response times delivered by the inclusion of non atomic mini-batches are marginal and that the running times of NFGS considerably overpass the execution time constraints typically used in the industry, we conclude that FGS is the best algorithm for practical applications of~\LTSP. 
	Preliminary tests in the same computational environment suggest that \NFGS and \LogNFGS are still satisfactory  in scenarios where the time limit is set to 1 second and the number of files is bounded by 16,000 and 25,000, respectively. 
	
	% Finally, we remark that neither the length of the time horizon nor the number of files had a significant impact on the  relative results.

	\subsubsection*{Landsat instances} 
	
	%\begin{table}
	%		\footnotesize
	%	\centering
	%	\caption{Performance of algorithms for \LTSP on Landsat instances.}
	%	\label{tab:offlandsat}
	%	\begin{tabular}{|c|cccccc|}
	%		\hline
	%		Instance   &    \SSS &     GS &    FGS & NFGS &   Log-NFGS &   LTFS+ \\
	%		\hline
	%		02-15-22   & 10.40 & 10.40 &  8.76 & \textbf{8.76} &       8.76 &   13.42 \\
	%		09-15-32   & 10.26 & 10.29 &  8.61 & \textbf{8.61} &       8.61 &   11.90 \\
	%		17-21-32   & 10.44 & 10.37 &  8.68 & \textbf{8.67} &       8.68 &   12.29 \\
	%		02-21-30   & 10.18 & 10.14 &  8.53 & \textbf{8.52} &       8.53 &   11.64 \\
	%		07-18-32   & 10.35 & 10.29 &  8.57 & \textbf{8.57} &       8.57 &   11.70 \\
	%		17-26-38   & 10.15 & 10.09 &  8.47 & \textbf{8.47} &       8.47 &   11.76 \\
	%		07-26-32   & 10.35 & 10.24 &  8.56 & \textbf{8.55} &       8.56 &   11.73 \\
	%		09-26-32   & 10.23 & 10.30 &  8.55 & \textbf{8.55} &       8.55 &   12.77 \\
	%		17-18-30   & 10.23 & 10.18 &  8.52 & \textbf{8.52} &       8.52 &   12.21 \\
	%		17-26-38   & 10.25 & 10.19 &  8.46 & \textbf{8.46} &       8.46 &   13.12 \\
	%		\hline
	%	\end{tabular}
	%\end{table}

	The leftmost part of Table~\ref{tab:landsat} contains the results of our experiments involving the Landsat instances.   Overall, the results are very similar to those observed for the Synthetic instances; \NFGS delivered the best results, closely followed by \FGS and \LogNFGS.  The performance of \SSS and \GS could not be distinguished in a statistically significant way (we have obtained a $p$-value of 0.23); each delivered better solutions than the other for 50 instances, and relative improvements were never larger than 2\%.  \FGS was better than \SSS and \GS in all cases, with average improvements of almost 17\%. \LogNFGS improved only one instance over \FGS (by less than 0.0003\%), whereas \NFGS improved 44 (yielding solutions that were 0.05\% better). Finally, \NFGS was 0.05\%  better than \LogNFGS in 43 cases.
	From the computational perspective, all algorithms were sufficiently fast  for practical applications (all executions were finished in less than 0.002 seconds).
	
	%and its running time was negligible for these instances, thus corroborating our previous observation supportung its that showing that it could be used in practical scenarios where the number of files is relatively small (i.e., less than 1000 files). Overall, the results for the Landsat instances are similar to those observed for the Synthetic dataset, although some relevant differences can also be observed. 

	%FGS, NFGS, and Log-NFGS clearly over-performed all the other strategies. In particular, FGS delivered better results than \SSS and GS in all cases, with solutions that were 17.1\% and 16.4\% better in average, respectively. The differences in performance between Log-NFGS and FGS  were negligible; the former delivered better results than the latter in only 5\% of the cases, with average improvements inferior to 0.0002\%. NFGS over-performed Log-NFGS in 55\% of the cases and the solutions differed in a statistically significant way ($p$-value of $1.1\times10^{-10}$), but improvements were inferior to 0.05\% in average. Finally, even the differences between NFGS and FGS were not remarkable: 57\% of the instances got improved, in average by less than 0.05\% as well.

	\subsection{Online algorithms}
	
	In our experiments with~\OLTSP, we tested \LTFS, \ELTFS,  \ReplanSSS, \ReplanGS, \ReplanFGS, and \ReplanLogNFGS. 
	%, which, in the case of Replan, are largely dominated by the time spent with the solution of~\LTSP subproblems. Objective values are also provided as the ratio between the overall response times produced by the indicated algorithm and the value yielded by LTFS. 
	Note  that the  dominance relationships between the strategies that were  analytically demonstrated for~\LTSP do not necessarily extend to their online \Replan counterparts, which reinforces the relevance of empirical comparisons.

	\subsubsection*{Synthetic instances} 
	
	%The results of our computational experiments involving synthetic instances of~\LTSP are presented in Table~\ref{tab:offsynthetic}. In addition to the relative performance of \SSS, GS, FGS, and NFGS with respect to LTFS, we also report the running time of NFGS; all the other algorithms had a negligible running time (less than 0.01 seconds) and can be considered acceptable for practical applications. In summary, all the algorithms proposed in this article clearly over-perform  LTFS; in particular, LTFS did no provide a better schedule than any of the algorithms for any of the 150 instances composing this dataset. NFGS delivered the better results, but FGS optimizes the trade-off between running time and solution quality and seems to be more adequate for practical scenarios, as discussed below. 
	
	The results for the Synthetic instances are presented in Table~\ref{tab:onsynthetic}. In addition to the relative response times (compared against \ELTFS), we also report the average simulation times of these algorithms, which include the time spent with the computation of the scheduling sub-routines. Differently from the results for the offline algorithms, the Synthetic instances show a diversity in terms of performance across the difference algorithms. In particular, $|\timehorizon|$ played an important role in these tests, as we discuss next.

	\ReplanGS was considerably worse than all the other algorithms, including \ELTFS, the baseline for our comparisons. \ReplanGS is very susceptible to minor events and completely ignores the global status of the system, that is, it services files even when they clearly should not given the number of pending requests for other files in the tape.
	
	\ReplanSSS performed considerably better than the others in scenarios where~$k = 1$. An explanation for this phenomenon is the way releases are distributed over time in these instances. Namely, intervals between requests associated with the same file are large enough to make the execution of several requests on Phase 2 not so expensive in terms of response times; i.e., there are many requests being released by the time the tape is reaching its leftmost position. At the same time, all releases are concentrated within time horizon~$|\timehorizon| = m$, so virtually all requests are expected to be available by the time Phase 2 is executed for the first time, so a second execution of Phase 1 is typically not necessary in these scenarios.

	In an aggregate analysis, though (i.e., if we consider all scenarios together), both \ReplanFGS and \ReplanLogNFGS over-perform  \ReplanSSS in more scenarios (100 vs 50 in both cases) and with average improvements of almost 3\% (also in both cases). For 118 instances, \ReplanFGS had better results than \ReplanLogNFGS, but average improvements were small (less than 0.03\%); the same happened for scenarios where \ReplanLogNFGS was better, with even smaller improvements (less than 0.004\%).
	
	Finally, the running times show that the differences in performance between \ReplanLogNFGS and \ReplanFGS were not as significant as they were in the offline experiments.  The main reason for this is that the average size of the~\LTSP instances being solved is relative small (half of the size of the actual instance), so the amortized computational cost is not so large.  Moreover, in certain cases, \ReplanLogNFGS was actually faster than \ReplanFGS, thus showing that longer planning times may be compensated by better schedules in these scenarios. Nevertheless, the results with this benchmark suggest that \ReplanFGS should probably be the strategy of choice for general-purpose applications, as it delivered the best results and had satisfactory running times.

	\begin{table}
		\footnotesize
		\centering
		\caption{Performance of algorithms for \OLTSP on synthetic instances.}
		\label{tab:onsynthetic}
		\begin{tabular}{|c|cccc|cccc|}
			\hline 
			\multicolumn{1}{c}{Instance}		& \multicolumn{4}{c}{Relative objetive values}  & \multicolumn{4}{c}{Running times (sec.)} \\
			& \SSS          &     GS & FGS          & Log-NFGS     &     \SSS &      GS &      FGS &   Log-NFGS \\
			\hline
			20000-1    & \textbf{70.65} & 113.07 & 74.91        & 74.97        &   1.908 &   0.999 &    6.873 &      7.916 \\
			20000-3    & 89.55        & 109.00    & \textbf{88.00} & 88.02        &   4.943 &   2.483 &   17.915 &     16.869 \\
			20000-5    & 89.80        & 119.53 & \textbf{86.18} & 86.19        &   7.869 &   4.293 &   27.951 &     25.939 \\
			40000-1    & \textbf{73.48} & 117.56 & 77.85        & 77.90        &   7.717 &   3.995 &   34.002 &     41.000     \\
			40000-3    & 81.50        &  99.28 & \textbf{80.07} & 80.08        &  20.619 &  10.252 &  100.404 &     90.064 \\
			40000-5    & 92.55        & 123.29 & \textbf{88.87} & 88.87        &  32.680  &  17.203 &  142.053 &    130.133 \\
			60000-1    & \textbf{73.79} & 117.97 & 78.19        & 78.20        &  27.364 &  11.376 &   79.377 &     98.233 \\
			60000-3    & 89.89        & 109.55 & \textbf{88.33} & 88.34        &  61.193 &  28.469 &  222.970  &    218.202 \\
			60000-5    & 93.94        & 125.00    & 90.17        & \textbf{90.17} & 107.411 &  57.615 &  350.643 &    327.371 \\
			80000-1    & \textbf{69.85} & 111.74 & 74.03        & 74.05        &  74.505 &  29.944 &  154.08  &    200.923 \\
			80000-3    & 81.39        &  99.36 & \textbf{80.01} & 80.01        & 155.482 &  68.600   &  445.981 &    432.341 \\
			80000-5    & 88.93        & 118.37 & 85.34        & \textbf{85.34} & 346.643 & 169.698 &  696.005 &    594.049 \\
			100000-1   & \textbf{69.71} & 111.59 & 73.90        & 73.91        &  92.131 &  44.766 &  244.724 &    313.554 \\
			100000-3   & 85.03        & 103.76 & \textbf{83.57} & 83.57        & 327.077 & 152.815 &  720.067 &    732.237 \\
			100000-5   & 90.17        & 120.02 & 86.54        & \textbf{86.54} & 607.318 & 306.56  & 1195.86  &   1155.000     \\
			\hline
		\end{tabular}
	\end{table}

	\subsubsection*{Landsat instances}

	Our results for the online Landsat instances are presented in the rightmost part of Table~\ref{tab:landsat}.  Direct visual inspection allows us to see that the relative performances of the algorithms were similar to the ones observed for the offline experiments.

	The performances of \ELTFS and \SSS were slightly improved in comparison with the offline experiments, whereas \GS became worse. \ELTFS remained worse than all the other algorithms, but \SSS got significantly better than \GS, delivering superior results in 81 cases (average improvements of 2.5\%). \LogNFGS was better  than \FGS in more instances (48 vs. 38), but the differences in performance between both algorithms was not statistically significant ($p$-value of 0.15).

	Finally, most release times for Landsat instances are concentrated in the very beginning of the planning horizon (i.e., release times are much smaller than~$m$). Consequently, these instances were actually different from the 
	$|\files|-1$ configurations composing the Synthetic dataset, which explains why the superiority of  \SSS over the other algorithms that have been observed for this benchmark was not replicated at a similar extent to the Landsat instances.

	%Overall, the simplest strategies (LTFS+, \SSS, GS) were able to significantly improve the schedules generated by LTFS, with \SSS being slightly better than GS in 74\% of these scenarios. \hl{LTFS+ and \SSS delivered again exactly the same results. An explanation for this fact is that all requests have been released within a relatively short time horizon and, by the time the leftmost file was reached, all requests were released, so the algorithms ended up having the same performance}.

	%FGS and Log-NFGS were both superior to \SSS in all cases, with an average improvement of more than~16\% in both cases. Log-NFGS was better than FGS in 62\% of the instances, with average improvements superior to 10\% in these cases, but FGS did produce schedules with smaller response times for a few cases (6 instances). 

	\begin{table}
		\footnotesize
		\centering
		\caption{Results for Landsat instances.}
		\label{tab:landsat}
		\begin{tabular}{|c|cccccc|ccccc|}
			\hline 
			\multicolumn{1}{c}{Instance}		& \multicolumn{6}{c}{Offline}  & \multicolumn{5}{c}{Online} \\
			&    \SSS &     GS &    FGS & NFGS &   Log-NFGS &   LTFS+ &    \SSS &     GS &    FGS & Log-NFGS      &   LTFS+ \\
			\hline 
			02-15-22   & 10.40 & 10.40 &  8.76 & \textbf{8.76} &       8.76 &   13.42 & 10.21 & 10.47 & 8.70        & \textbf{8.70} &   13.23\\
			09-15-32   & 10.26 & 10.29 &  8.61 & \textbf{8.61} &       8.61 &   11.90 & 10.12 & 10.38 & \textbf{8.58} & 8.59        &   11.76 \\
			17-21-32   & 10.44 & 10.37 &  8.68 & \textbf{8.67} &       8.68 &   12.29 & 10.31 & 10.43 & \textbf{8.64} & 8.64        &   12.16 \\
			02-21-30   & 10.18 & 10.14 &  8.53 & \textbf{8.52} &       8.53 &   11.64 & 10.05 & 10.18 & 8.47        & \textbf{8.47} &   11.52\\
			07-18-32   & 10.35 & 10.29 &  8.57 & \textbf{8.57} &       8.57 &   11.70 & 10.22 & 10.33 & 8.52        & \textbf{8.51} &   11.56\\
			17-26-38   & 10.15 & 10.09 &  8.47 & \textbf{8.47} &       8.47 &   11.76 & 10.02 & 10.18 & 8.45        & \textbf{8.43} &   11.62\\
			07-26-32   & 10.35 & 10.24 &  8.56 & \textbf{8.55} &       8.56 &   11.73 & 10.19 & 10.32 & \textbf{8.51} & 8.52        &   11.57\\
			09-26-32   & 10.23 & 10.30 &  8.55 & \textbf{8.55} &       8.55 &   12.77 & 10.09 & 10.35 & 8.50        & \textbf{8.50} &   12.63\\
			17-18-30   & 10.23 & 10.18 &  8.52 & \textbf{8.52} &       8.52 &   12.21 & 10.12 & 10.30 & 8.52        & \textbf{8.49} &   12.10\\
			17-26-38   & 10.25 & 10.19 &  8.46 & \textbf{8.46} &       8.46 &   13.12 & 10.12 & 10.34 & 8.46        & \textbf{8.43} &   12.99\\
			\hline
		\end{tabular}
	\end{table}

	%\begin{table}
	%	\footnotesize
	%	\centering
	%	\caption{Performance of algorithms for \OLTSP on Landsat instances.}
	%	\label{tab:onlandsat}
	%	\begin{tabular}{|c|ccccc|}
	%		\hline
	%	Instance   &    \SSS &     GS &    FGS & Log-NFGS      &   LTFS+ \\
	%	\hline
	%	02-15-22   & 10.21 & 10.47 & 8.70        & \textbf{8.70} &   13.23 \\
	%	09-15-32   & 10.12 & 10.38 & \textbf{8.58} & 8.59        &   11.76 \\
	%	17-21-32   & 10.31 & 10.43 & \textbf{8.64} & 8.64        &   12.16 \\
	%	02-21-30   & 10.05 & 10.18 & 8.47        & \textbf{8.47} &   11.52 \\
	%	07-18-32   & 10.22 & 10.33 & 8.52        & \textbf{8.51} &   11.56 \\
	%	17-26-38   & 10.02 & 10.18 & 8.45        & \textbf{8.43} &   11.62 \\
	%	07-26-32   & 10.19 & 10.32 & \textbf{8.51} & 8.52        &   11.57 \\
	%	09-26-32   & 10.09 & 10.35 & 8.50        & \textbf{8.50} &   12.63 \\
	%	17-18-30   & 10.12 & 10.30 & 8.52        & \textbf{8.49} &   12.10  \\
	%	17-26-38   & 10.12 & 10.34 & 8.46        & \textbf{8.43} &   12.99 \\
	%	\hline
	%	\end{tabular}
	%\end{table}

	\section{Conclusion}\label{sec:conclusion}
	
	We investigated in this article the Linear Tape Scheduling Problem, which aims to identify strategies for the execution of read operations in single-tracked magnetic tapes such that the overall response times are minimized. \LTSP is similar to  classic combinatorial optimization problems, such as the Traveling Repairmen Problem and the Dial-a-Ride Problem in a line, but peculiarities on the behavior of magnetic tapes make \LTSP an interesting problem on its own.
	
	In practical settings, the Linear Tape Scheduling Problem is  online, so in addition to the basic version of the problem, we also investigated~\LTSPR, the extension with release times, and~\OLTSP, the online extension of~\LTSPR. Several theoretical results have been presented, including structural properties, complexity and competitive analysis, and algorithms. In particular, we showed that~\LTSPR is NP-hard, \OLTSP does not admit $c$-competitive algorithms for any given constant~$c$, and that~\LTSP can be solved efficiently if all files have the same size and the number of requests per file is at most one; we conjecture that the general case of~\LTSP   is NP-hard.%, but the complexity of the problem remains open. 

	Practical applications of the Linear Tape Scheduling Problem and its variants impose very short time limits for scheduling decisions, and this condition strongly influenced the design of the algorithms presented in this work. For~\LTSP, we presented 3-approximation schemes  which are sufficiently fast for real-world scenarios. 
	%We also developed exact and efficient algorithms for special cases of the problem. 
	We also proposed an online algorithm for~\OLTSP of satisfactory performance that significantly over-performed the strategy being currently used in the industry, as we have shown in our computational experiments on synthetic and real-world datasets.

	Finally, there are several interesting research directions motivated by the Linear Tape Scheduling Problem that could be investigated in future work. From the theoretical standpoint, some  special cases of~\LTSP seem to be as hard as the original problem itself, but these questions remain open; some examples are scenarios where all files have the same size and the number of requests per file is arbitrary and where files may have different sizes and at most one request each. A natural extension of~\LTSP with clear practical relevance involves multi-track serpentine tapes, which may consist of two or more tracks with potentially  different reading directions.

	\section*{Acknowledgment} 
	The authors would like to thank Andre Cire for interesting discussions, comments,  and suggestions.

	%\begin{quote}
	%\begin{small}
	\bibliographystyle{plainnat}
	\bibliography{references}
	%\end{small}
	%\end{quote}

\end{document}